\documentclass[sigconf]{acmart}

\renewcommand\footnotetextcopyrightpermission[1]{} 
\AtBeginDocument{%
  \providecommand\BibTeX{{%
    \normalfont B\kern-0.5em{\scshape i\kern-0.25em b}\kern-0.8em\TeX}}}

\setcopyright{acmcopyright}
\copyrightyear{2018}
\acmYear{2018}
\acmDOI{10.1145/1122445.1122456}

\usepackage{nicefrac}       
\usepackage{microtype}      
\usepackage{graphicx}
\graphicspath{{../Figures/}}
\usepackage{booktabs} 
\usepackage{caption}
\usepackage{subcaption}
\usepackage{algorithm}
\usepackage{algpseudocode}

\usepackage{amsopn,amssymb,amsthm,amsmath,bm}
\usepackage{paralist}
\usepackage{multirow}
\usepackage{multicol}
\usepackage{float}
\usepackage{thmtools}
\usepackage{thm-restate}
\usepackage{cleveref}
\usepackage{framed}
\usepackage{hyperref}
\usepackage{enumitem}
\usepackage{bm}
\usepackage{colortbl}

\usepackage{IEEEtrantools}
\allowdisplaybreaks

\newtheorem{claim}{Claim}
\newtheorem{remark}{Remark}

\usepackage{bm}
\usepackage{colortbl}

\newcommand{\rank}{\textrm{rank}}

\newcommand{\Z}{\mathbb{Z}} 
\newcommand{\GF}[1]{\mathbb{F}_{#1}} 

\newcommand{\dsample}[2]{{#1}^{(\downarrow #2)}} 

\newcommand{\shift}[2]{{#1}^{(#2)}} 

\newcommand{\set}[1]{\mathcal{#1}}
\newcommand{\vect}[1]{\mathbf{#1}}

\newcommand{\rand}[1]{\mathbf{#1}}

\newcommand{\w}[1]{\mathbf{w}^{#1}}

\newcommand{\wloc}[3]{{\mathbf{w}}^{#1}_{#2,#3}}

\newcommand{\gradg}[3]{g^{#1}_{#2,#3}}

\newcommand{\Expsub}[2]{\mathbb{E}_{#1}\left[#2\right]}

\newcommand{\etal}{\eta_l}

\newcommand{\bigOh}[1]{\mathcal{O}\left(#1\right)}

\newcommand{\kpk}[1]{\mathsf{pk}_{#1}}
\newcommand{\ksk}[1]{\mathsf{sk}_{#1}}
\newcommand{\cbrack}[1]{\left\{#1\right\}}
\newcommand{\sh}[2]{\mathsf{sh}_{#1\rightarrow#2}}
\newcommand{\shi}[1]{\mathsf{sh}_{#1}}
\newcommand{\esh}[2]{c_{#1 \rightarrow #2}}

\newcommand{\nsubi}[1]{n_{#1}}

\newcommand{\proc}[1]{\textsc{#1}} 
\newcommand{\share}[2]{\left[#1\right]_{#2}} 

\newcommand{\Hcond}[2]{H\left(#1\mid#2\right)} 
\newcommand{\Hp}[1]{H\left(#1\right)} 

\newcommand\blfootnote[1]{%
  \begingroup
  \renewcommand\thefootnote{}\footnote{#1}%
  \addtocounter{footnote}{-1}%
  \endgroup
}

\begin{document}

\title{\textsc{FastSecAgg}: Scalable Secure Aggregation for Privacy-Preserving Federated Learning}

\author{Swanand Kadhe, Nived Rajaraman, O. Ozan Koyluoglu, and Kannan Ramchandran}
\affiliation{%
    \department{Department of Electrical Engineering and Computer Sciences}
  \institution{University of California Berkeley}
}
\email{{swanand.kadhe, nived, ozan.koyluoglu, kannanr}@berkeley.edu}

\renewcommand{\shortauthors}{Kadhe et al.}

\begin{abstract}
Recent attacks on federated learning demonstrate that keeping the training data on clients' devices does not provide sufficient privacy, as the model parameters shared by clients can leak information about their training data. A \emph{secure aggregation} protocol enables the server to aggregate clients' models in a privacy-preserving manner. However, existing secure aggregation protocols incur high computation/communication costs, especially when the number of model parameters is larger than the number of clients participating in an iteration -- a typical scenario in federated learning. 

In this paper, we propose a secure aggregation protocol, \proc{FastSecAgg}, that is efficient in terms of computation and communication, and robust to client dropouts. The main building block of \proc{FastSecAgg} is a novel  multi-secret sharing scheme, \textsc{FastShare}, based on the Fast Fourier Transform (FFT), which may be of independent interest. 
\proc{FastShare} is information-theoretically secure, and achieves a trade-off between the number of secrets, privacy threshold, and dropout tolerance.
Riding on the capabilities of \proc{FastShare}, we prove that \proc{FastSecAgg} is (i) secure against the server colluding with {\it any} subset of some constant fraction (e.g. $\sim10\%$) of the clients in the honest-but-curious setting; and (ii) tolerates dropouts of a {\it random} subset of some constant fraction (e.g. $\sim10\%$) of the clients.
\proc{FastSecAgg} achieves significantly smaller computation cost than existing schemes while achieving the same (orderwise) communication cost. 
In addition, it guarantees security against adaptive adversaries, which can perform client corruptions dynamically during the execution of the protocol.
\end{abstract}



\maketitle

\blfootnote{A shorter version of this paper has been accepted in ICML Workshop on Federated Learning for User Privacy and Data Confidentiality, July  2020, and CCS Workshop on Privacy-Preserving Machine Learning in Practice, November 2020.
}

\section{Introduction}
\label{intro}
Federated Learning (FL) is a distributed learning paradigm that enables a large number of clients to coordinate with a central server to learn a shared model while keeping all the training data on clients' devices. 
In particular, training a neural net in a federated manner is an iterative process. 
In each iteration of FL, the server  selects a subset of clients, and sends them the current global model. 
Each client runs several steps of mini-batch stochastic gradient descent, and communicates its model update---the difference between the local model and the received global model---to the server. 
The server aggregates the model updates from the clients to obtain an improved global model, and moves to the next iteration.

FL ensures that every client keeps their local data on their own device, and sends only model updates that are functions of their dataset. Even though any update cannot contain more information than the client's local dataset, the updates may still leak significant information making them vulnerable to inference and inversion attacks~\cite{Nasr:19:SP,Fredrikson:15:CCS,Shokri:17:SP,Ganju:18:CCS}. 
Even well-generalized deep models such as DenseNet can leak a significant amount of information about their training data~\cite{Nasr:19:SP}.
In fact, certain neural networks (e.g., generative text models) trained on sensitive data (e.g., private text messages) can  memorize the training data~\cite{Carlini:19:Usenix}.

{\it Secure aggregation} protocols can be used to provide strong privacy guarantees in FL. At a high level, secure aggregation is a secure multi-party computation (MPC) protocol that enables the server to compute the sum of clients' model updates without learning any information about any client's individual update~\cite{Bonawitz:SecAgg:17}. Secure aggregation can be efficiently composed with differential privacy to further enhance privacy guarantees~\cite{Bonawitz:SecAgg:17,Goryczka:17,Truex:19:CCS}.

The secure aggregation problem for securely computing the summation has received significant research attention in the past few years, see \cite{Goryczka:17} for a survey. However, the FL setup presents unique \textit{challenges} for a secure aggregation protocol~\cite{Bonawitz:19:scale,Kairouz:19:survey}:
\begin{enumerate}
    \item \textit{Massive scale:} 
    Up to 10,000 users can be selected to participate in an iteration.
    \item \textit{Communication costs:} Model updates can be high-dimensional vectors, e.g., the ResNet model consists of hundreds of millions of parameters~\cite{He:16:ResNet}. 
    \item  \textit{Dropouts:} Clients are mobile devices which can drop out at any point. The percentage of dropouts can go up to 10\%.
    \item \textit{Privacy:} It is  crucial to provide strongest possible privacy guarantees since malicious servers can launch powerful attacks by using \textit{sybils}~\cite{Douceur:02:sybil}, wherein the adversary simulates a large number of fake client devices~\cite{Fung:19:FLsybil}.
\end{enumerate}

\begin{table*}[!t]
\vskip 0.15in
\begin{center}
    {\small
    \begin{tabular}{|l||c|c|c|c|c|c|}
    \hline
    {Protocol} & 
     \begin{tabular}{@{}c@{}}Computation \\ (Server) \end{tabular}   &
    \begin{tabular}{@{}c@{}}Communication \\ (Server) \end{tabular}   &
    \begin{tabular}{@{}c@{}}Computation \\ (Client) \end{tabular}   &
    \begin{tabular}{@{}c@{}}Communication \\ (Client) \end{tabular} & 
    Adversary & 
    Dropouts\\
    \hline
    \hline
    \textsc{SecAgg} ~\cite{Bonawitz:SecAgg:17} & 
    $\bigOh{L N^2}$ & 
    $\bigOh{LN + N^2}$ & 
    $\bigOh{LN + N^2}$ & 
    $\bigOh{L + N}$ &
    Adaptive &
    Worst-case\\
    \hline
    \textsc{TurboAgg} ~\cite{So:20:TurboAgg} &
    $\bigOh{L\log N \log^2\log N}$ &
    $\bigOh{LN\log N}$ &
    $\bigOh{L\log N \log^2\log N}$ &
    $\bigOh{L\log N}$ &
    Non-adaptive &
    Average-case\\
    \hline
    \textsc{SecAgg+} ~\cite{Bell:20:SecAggPlus} &
    $\bigOh{LN \log N + N \log^2 N}$ &
    $\bigOh{LN + N\log N }$ &
    $\bigOh{L\log N  + \log^2 N}$ &
    $\bigOh{L + \log N }$ &
    Non-adaptive &
    Average-case\\
    \hline
    \rowcolor{gray!30}
    \textsc{FastSecAgg} &
    ${\bigOh{L\log N}}$ &
    ${\bigOh{LN + N^2}}$ &
    $\bigOh{L\log N}$ &
    ${\bigOh{L + N}}$ &
    Adaptive &
    Average-case\\
    \hline
    \end{tabular}
    }
    \vspace{1mm}
    \caption{Comparison of the proposed \textsc{FastSecAgg} with \textsc{SecAgg}~\cite{Bonawitz:SecAgg:17}, \textsc{TurboAgg}~\cite{So:20:TurboAgg}, and \textsc{SecAgg+}~\cite{Bell:20:SecAggPlus}. Here $N$ is the total number of clients and $L$ is the length of model updates (i.e., client inputs). 
    An adaptive adversary can choose which clients to corrupt during the protocol execution, whereas the corruptions happen before the protocol execution starts for a non-adaptive adversary. 
    Worst-case (resp. average-case) dropout guarantee ensures that any (resp. a random) subset of clients of a bounded size may drop out without affecting the correctness and security. 
    \proc{TurboAgg} requires at least $\log N$ rounds, which increases the overheads. 
    The protocols in \cite{Truex:19:CCS, Xu:19:hybridalpha} assume a trusted third party for key distribution (see Sec.~\ref{sec:comparison} for details). 
    \label{tbl:comparison}
    }
    \end{center}
\end{table*}

Due to these challenges, existing secure aggregation protocols for FL \cite{Bonawitz:SecAgg:17,Truex:19:CCS,Xu:19:hybridalpha,So:20:TurboAgg,Bell:20:SecAggPlus} either incur heavy computation and communication costs or provide only weaker forms of privacy guarantees, which limits their scalability (see Table~\ref{tbl:comparison}). 

Secret sharing \cite{Beimal:11} is an elegant primitive that is used as a building block in a large number of secure summation (and in general secure MPC) protocols, see, e.g., \cite{Goryczka:17,Evans:18:now}. Secret sharing based secure aggregation protocols are often more efficient than other approaches such as homomorphic encryption and Yao's garbled circuits \cite{Burkhart:10:sepia,Acs:11:dream,Goryczka:17}. However, 
na\"ively using the popular Shamir's secret sharing scheme \cite{Shamir:79} for aggregating length-$L$ updates from $N$ clients incurs 
$\bigOh{LN^2}$ communication and computation cost, which is intractable at the envisioned scale of FL.  

In this paper, we propose \proc{FastShare}---a novel secret sharing scheme based on a finite-field version of the Fast Fourier Transform (FFT)---that is designed to leverage the FL setup. 
\textsc{FastShare} is a \textit{multi-secret} sharing scheme, which can {\it simultaneously} share $S$ secrets to $N$ clients such that (i) it is information-theoretically secure against a collusion of \textit{any} subset of some constant fraction (e.g. $\sim 10\%$) of the clients; and (ii) with overwhelming probability, it is able to recover the secrets with $\bigOh{N\log N}$ complexity even if a \textit{random} subset of some constant fraction (e.g. $\sim 10\%$) of the clients drop out. 
{In general, \proc{FashShare} can achieve a trade-off between the security and dropout parameters (see Theorem~\ref{thm:main}); thresholds of $10\%$ are chosen to be well-suited in the FL setup.} 
We note that a larger number of  \textit{arbitrary} (or even adversarial) dropouts does not affect security, and may only result in a failure to recover the secrets (i.e, affects only correctness).
\proc{FastShare} uses ideas from sparse-graph codes~\cite{Luby:98,Richardson:01,Justesen:11} and crucially exploits \textit{insights derived from a spectral view of codes}~\cite{Blahut:79,Pawar:18}, which enables us to prove strong privacy guarantees.

Using \proc{FastShare} as a building block, we propose \proc{FastSecAgg} -- a scalable secure aggregation protocol for FL that has low computation and communication costs and strong privacy guarantees (see Table~\ref{tbl:comparison} for a comparison of costs).
\proc{FastSecAgg} achieves low costs by riding on the capabilities of \proc{FastShare} and leveraging the following unique attribute of FL.
Privacy breaches in FL are inflicted by malicious/adversarial actors (e.g., a server using sybils), requiring \textit{worst-case} guarantees; whereas dropouts occur due to non-adversarial/natural causes (e.g., system dynamics and wireless outages)~\cite{Bonawitz:19:scale}, making \textit{statistical or average-case} guarantees sufficient. 
\proc{FastShare} achieves  $\bigOh{N\log N}$ computation cost by relaxing the dropout constraint from worst-case to random.

For $N$ clients, each having a length-$L$ update vector, \proc{FastSecAgg} requires $\bigOh{L \log N}$ computation and $\bigOh{LN + N^2}$ communication at the server, and $\bigOh{L \log N}$ computation and $\bigOh{L + N}$ communication per client. We compare the costs and setups with existing protocols in Table~\ref{tbl:comparison} (see Sec.~\ref{sec:comparison} for details). We point out that the computation cost at the server was observed to be the key barrier to scalability in~\cite{Bonawitz:SecAgg:17,Bonawitz:19:scale}. \proc{FastSecAgg} significantly reduces the computation cost at the server as compared to prior works. 
When the length of model updates $L$ is larger than the number of clients $N$, \proc{FastSecAgg} achieves the same (order-wise) communication cost as the prior works. 
We note that, in typical FL applications, the number of model parameters $(L)$ is in millions, whereas the number of clients $(N)$ is up to tens of thousands~\cite{McMahan:17,Bonawitz:19:scale}. 
In addition, \proc{FastSecAgg} guarantees security against an {\it adaptive} adversary, which can corrupt clients adaptively during the execution of the protocol. This is in contrast with the protocols in~\cite{So:20:TurboAgg,Bell:20:SecAggPlus}, which are secure only against non-adaptive (or static) adversaries, wherein client corruptions happen before the protocol executions begins.



We note that, although \proc{FastShare} secret sharing scheme is inspired from the federated learning setup, it may be of independent interest. In fact, \proc{FastShare} is a cryptographic primitive that can be used for secure aggregation (and secure multi-party computation in general) in wide applications including smart meters~\cite{Smart-meter:11}, medical data collection~\cite{Health-data:10}, and syndromic surveillance~\cite{Biosurveillance:06}.

\section{Problem Setup and Preliminaries}
\label{sec:problem-setup}

\subsection{Federated Learning}
\label{sec:federated-learing}
We consider the setup in which a fixed set of $M$ clients, each having their local dataset, coordinate with the server to jointly train a model. 
The $i$-th client's data is sampled from a distribution $\set{D}_i$. 
The federated learning problem can be formalized as minimizing a sum of stochastic functions defined as
\begin{equation}
    \label{eq:federated-learning}
    \arg\min_{\mathbf{w}\in\mathbb{R}^d} \left\{\ell(\mathbf{w}) = \frac{1}{M}\sum_{i=1}^{M}\ell_{i}(\vect{w}) 
    \right\},
\end{equation}
where $\ell_{i}(\vect{w}) = \Expsub{\zeta \sim \mathcal{D}_i}{\ell_{i}(\vect{w};\zeta)}$ is the expected loss of the prediction on the $i$-th client's data made with model parameters $\vect{w}$.

Federated Averaging (\textsc{FedAvg}) is a synchronous update scheme that proceeds in rounds of communication~\cite{McMahan:17}. At the beginning of each round (called iteration), the server selects a subset $\set{C}$ of $N$ clients (for some $N\leq M$). Each of these clients $i\in\set{C}$ copies the current model parameters $\wloc{t}{i}{0} = \w{t}$, and performs $T$ steps of (mini-batch) stochastic gradient descent steps to obtain its local model $\wloc{t}{i}{T}$; each local step $k$ is of the form $\wloc{t}{i}{k} \leftarrow \wloc{t}{i}{k-1} - \etal \gradg{t}{i}{k-1}$, where $\gradg{t}{i}{k-1}$ is an unbiased stochastic gradient of $\ell_{i}$ at $\wloc{t}{i}{k-1}$, and $\etal$ is the local step-size.\footnote{In practice, clients can make multiple training passes (called epochs) over its local dataset with a given step size. Further, typically client datasets are of different size, and the server takes a weighted average with weight of a client proportional to the size of its dataset. See~\cite{McMahan:17} for details.}
Then, each client $i\in\set{C}$ sends their update as $\Delta{\vect{w}}^{t}_{i}=\wloc{t}{i}{L} - \w{t}$. At the server, the clients' updates $\Delta{\vect{w}}^{t}_{i}$ are aggregated to form the new server model as $\w{t+1} = \w{t} + \frac{1}{|\set{C}|}\sum_{i\in\set{C}}\Delta\vect{w}^{t}_{i}$.

Our focus is on one iteration of \proc{FedAvg} and we omit the explicit dependence on the iteration $t$ hereafter. We assume that each client potentially compresses and suitably quantizes their model update $\Delta\vect{w}^t_i\in\mathbb{R}^d$ to obtain $\vect{u}_i\in\mathbb{Z}_R^L$, where $L \leq d$.
Our goal is to design a protocol that enables the server to securely compute $\sum_{i\in\set{C}}\vect{u}_i$. 
We describe the threat model and the objectives in the next section.

\subsection{Threat Model}
\label{sec:privacy-model}
Federated learning can be considered as a multi-party computation consisting of $N$ parties (i.e., the clients), each having their own private dataset, and an aggregator (i.e., the server), with the goal of learning a model using all of the datasets. 

\vspace{2pt}
\noindent\textbf{Honest-but-curious model:} The parties honestly follow the protocol, but attempt to learn about the model updates from other parties by using the messages exchanged during the execution of the protocol. The honest-but-curious (aka semi-honest) adversarial model is commonly used in the field of secure MPC, including prior works on secure federated learning~\cite{Bonawitz:SecAgg:17,Truex:19:CCS, So:20:TurboAgg}. 

\vspace{2pt}
\noindent\textbf{Colluding parties:} In every iteration of federated averaging, the server first samples a set $\set{C}$ of clients. The server may collude with any set of up to $T$ clients from $\set{C}$. 
The server can view the internal state and all the messages received/sent by the clients with whom it colludes. 
We refer to the internal state along with the messages received/sent by colluding parties (including the server) as their \textit{joint view}. (see Sec.~\ref{sec:security-analysis} for details)

\vspace{2pt}
\noindent\textbf{Dropouts:} A random subset of up to $D$ clients may drop out at any point of time during the execution of secure aggregation. 

\vspace{2pt}
\noindent\textbf{Objective:} Our objective is to design a protocol to securely aggregate clients' model updates such that the {joint view} of the server and \textit{any} set of up to $T$ clients must not leak any information about the other clients' model updates, besides what can be inferred from the output of the summation.
In addition, even if a \textit{random} set of up to $D$ clients drop out during an iteration, the server with high probability should be able to compute the sum of model updates (of the surviving clients) while maintaining the privacy. We refer to $T$ as the \textit{privacy threshold} and $D$ as the \textit{dropout tolerance}. 

\subsection{Cryptographic Primitives}
\label{sec:preliminaries}

\subsubsection{Key Agreement} 
\label{sec:key-agreement}
A key agreement protocol consists of three algorithms
(\proc{KA.param}, \proc{KA.gen}, \proc{KA.agree}). Given a security parameter $\lambda$, the parameter generation algorithm $pp \leftarrow \proc{KA.param}(\lambda)$ generates some public parameters, over which the protocol will be parameterized.
The key generation algorithm allows a client $i$ to generate a private-public key pair $(\kpk{i},\ksk{i}) \leftarrow \proc{KA.gen}(pp)$. The key agreement procedure allows clients $i$ and $j$ to obtain a private shared key $k_{i,j}\leftarrow\proc{KA.agree}(\mathsf{sk}_i, \mathsf{pk}_j)$.  
Correctness requires that, for any key pairs generated by clients $i$ and $j$ (using \proc{KA.gen} with the same parameters $pp$), $\proc{KA.agree}(\mathsf{sk}_i,\mathsf{pk}_j) = \proc{KA.agree}(\mathsf{sk}_j,\mathsf{pk}_i)$. 
Security requires that there exists a simulator $\mathsf{Sim}_{\proc{KA}}$, which takes as input an output key sampled uniformly at random and the public key of the other client, and simulates the messages of the key agreement execution such that the simulated messages are computationally indistinguishable from the protocol transcript.

\subsubsection{Authenticated Encryption} 
\label{sec:encryption}
An authenticated encryption allows two parties to communicate with data confidentially and data integrity. It consists of an encryption algorithm $\proc{AE.enc}$ that takes as input a key and a message and outputs a ciphertext, and a decryption algorithm $\proc{AE.dec}$ that takes as input a ciphertext and a key and outputs the original plaintext, or a special error symbol $\perp$. For correctness, we require that for all keys $k_i\in\{0,1\}^{\lambda}$ and all messages $m$,  $\proc{AE.dec}\left((k_i,\proc{AE.enc}(k_i,m)\right) = m$.
For security, we require semantic security under a chosen plaintext attack (IND-CPA) and ciphertext integrity (IND-CTXT) \cite{Bellare:00:crypto}.

\section{\textsc{FastShare}: FFT Based Secret Sharing}
\label{sec:F-Sec}
In this section, we present a novel, computationally efficient multi-secret sharing scheme \textsc{FastShare} which forms the core of our proposed secure aggregation protocol \proc{FastSecAgg} (described in the next section). 

A multi-secret sharing scheme splits a set of secrets into shares (with one share per client) such that coalitions of clients up to certain size have no information on the secrets, and a random\footnote{Secret sharing~\cite{Beimal:11} and multi-secret sharing~\cite{Franklin-Yung:92,Blundo:94} schemes conventionally consider \textit{worst-case} dropouts. We consider \textit{random} dropouts to exploit the FL setup.} set of clients of large enough size can jointly reconstruct the secrets from their shares. In particular, we consider information-theoretic (perfect) security. A secret sharing scheme is said to be linear if any linear combination of valid share-vectors result in a valid share-vector of the linear combination applied to the respective secret-vectors. We summarize this in the following definition.

\begin{definition}[Multi-secret Sharing]
\label{def:secret-sharing} 
Let $\GF{q}$ be a finite field, and let $S$, $T$, $D$ and $N$ be positive integers such that $S+T+D < N\leq q$. A linear multi-secret sharing scheme over $\GF{q}$ consists of two algorithms \proc{Share} and \proc{Recon}. The sharing algorithm  $\{(i,\share{\vect{s}}{i})\}_{i\in\set{C}}\leftarrow\proc{Share}(\vect{s},\set{C})$ is a probabilistic algorithm that takes as input a set of secrets $\vect{s}\in\GF{q}^{S}$ and a set $\set{C}$ of $N$ clients (client-IDs), and produces a set of $N$ shares, each in $\GF{q}$, where share $\share{\vect{s}}{i}$ is assigned to client $i$ in $\set{C}$. For a set $\set{D}\subseteq\set{C}$, the reconstruction algorithm $\{\vect{s}, \perp\}\leftarrow\proc{Recon}\left(\{(i,\share{\vect{s}}{i})\}_{i\in\set{C}\setminus\set{D}}\right)$ takes as input the shares corresponding to $\set{C}\setminus\set{D}$, and outputs either a set of $S$ field elements $\vect{s}$ or a special symbol $\perp$. The scheme should satisfy the following requirements.
\begin{enumerate}
    \item $T$-Privacy: For all $\vect{s},\vect{s}'\in\GF{q}^{S}$ and every $\set{P}\subset\set{C}$ of size at most $T$, the shares of $\vect{s}$ and $\vect{s}'$ restricted to $\set{P}$ are identically distributed. 
    \item $D$-Dropout-Resilience: For every $\vect{s}\in\GF{q}^{S}$ and any random set $\set{D}\subset\set{C}$ of size at most $D$, 
    $\proc{Recon}\left(\{(i,\share{\vect{s}}{i})\}_{i\in\set{C}\setminus\set{D}}\right) = \vect{s}$ with probability at least $1 - \text{poly}(N)$.
    \item \textit{Linearity:} If  $\{(i,\share{\vect{s}_1}{i})\}_{i\in\set{C}}$ and $\{(i,\share{\vect{s}_2}{i})\}_{i\in\set{C}}$ are sharings of $\vect{s}_1$ and $\vect{s}_2$, then  $\{(i,a\share{\vect{s}_1}{i}+b\share{\vect{s}_2}{i})\}_{i\in\set{C}}$ is a sharing of $a\vect{s}_1+b\vect{s}_2$ for any $a,b\in\GF{q}$.
\end{enumerate}
\end{definition}


Next, we describe \proc{FastShare} which can achieve a trade-off between $S$, $T$, and $D$ for a given $N$.
We begin with setting up the necessary notation. 
\proc{FastShare} leverages the finite-field Fourier transform and Chinese Remainder Theorem. Towards this end, let $\{n_0,n_1\}$ be co-prime positive integers of the same order, such that $n_0n_1$ divides $(q-1)$. 
Without loss of generality, assume that $n_0 < n_1$. 
E.g., $n_0 = 10$, $n_1 = 13$, and $q = 131$. 
We consider $N$ of the form $N = n_0n_1$, and choose the field size $q$ as a power of a prime such that $N$ divides $q-1$. 
By the co-primeness of $n_0$ and $n_1$, applying the Chinese Remainder Theorem, any number $j \in \{0,\cdots,n-1 \}$ can be uniquely represented in a 2-dimensional (2D) grid as a tuple $(a,b)$ where $a = j \mod n_0$ and $b = j \mod n_1$. 


\begin{figure}[!t]
    \centering
    \includegraphics[scale=0.45]{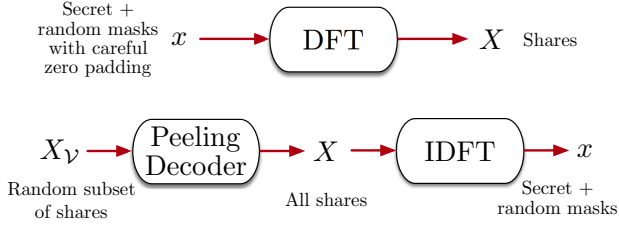}
    \caption{\proc{FastShare} to generate shares, and \proc{FastRecon} to recover the secrets from a subset of shares.}
    \label{fig:SS-share}
\end{figure}

\subsection{\proc{FastShare} to Generate Shares}
\label{sec:Fast-Share}
The sharing algorithm  $\{(i,\share{\vect{s}}{i})\}_{i\in\set{C}}\leftarrow\proc{FastShare}(\vect{s},\set{C})$ is a probabilistic algorithm that takes as input a set of secrets $\vect{s}\in\GF{q}^{S}$, and a set $\set{C}$ of $N$ clients (client-IDs), and produces a set of $N$ shares, each in $\GF{q}$, where share $\share{\vect{s}}{i}$ is assigned to client $i$.
At a high level, \proc{FastShare} consists of two stages. First, it constructs a length-$N$ ``signal'' consisting of the secrets and random masks with zeros placed at judiciously chosen locations. Second, it takes the fast Fourier transform of the signal to generate the shares (see Fig.~\ref{fig:SS-share}). The shares can be considered as the ``spectrum'' of the signal constructed in the first stage.

\begin{figure}[!t]
    \centering
    \includegraphics[width=0.475\textwidth]{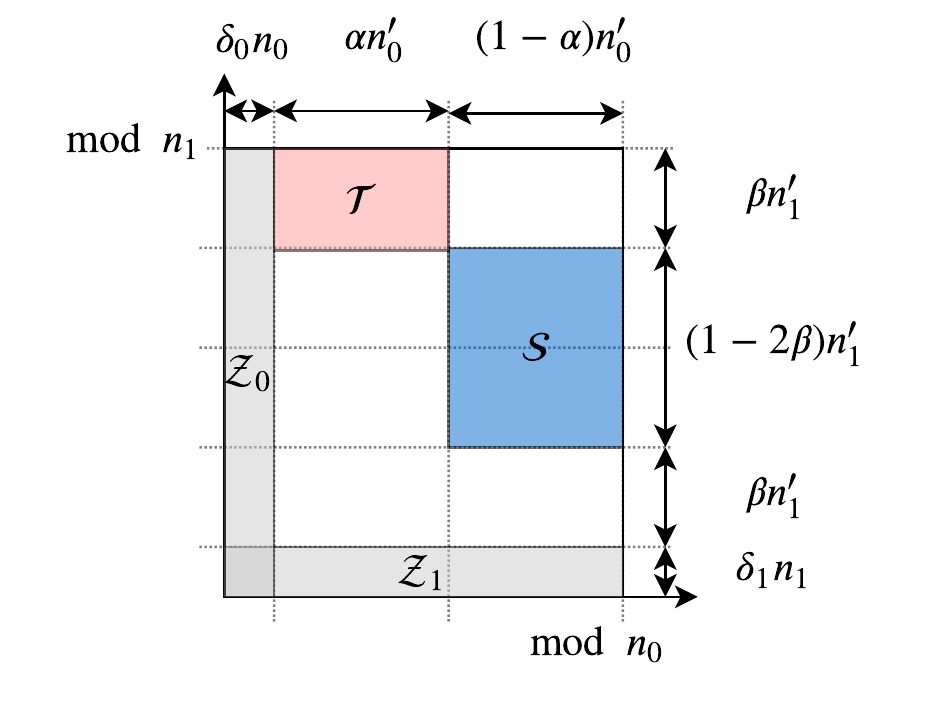}
    \caption{Indices assigned to a grid using the Chinese remainder theorem, and sets $\set{Z}_0$, $\set{Z}_1$, $\set{S}$, and $\set{T}$. Here, we define $n_i' = (1-\delta_i)n_i$ for $i\in\{0,1\}$. 
    Set $\set{S}$ is used for secrets, and sets $\set{Z}_0$ and $\set{Z}_1$ are used for zeros, and the remaining locations are used for random masks. Set $\set{T}$ is used in our privacy proof. 
    }
    \label{fig:2D-secrets}
\end{figure}

In particular, the ``signal'' is constructed as follows. 
Given fractions $\delta_0, \delta_1, \alpha \in (0,1)$ and $\beta\in(0,1/2)$, we define the sets of tuples $\set{Z}_0$, $\set{Z}_1$, $\set{S}$, and $\set{T}$ using the grid representation as depicted in  Fig.~\ref{fig:2D-secrets}. We give the formal definitions below. Here, we round any real number to the largest integer no larger than it, i.e., round $x$ to $\lfloor x\rfloor$, and omit the floor sign for the sake of brevity. 
\begin{IEEEeqnarray}{rCl}
    \label{eq:set-Z-0}
    \set{Z}_0 &=& \left\{(a,b) : 0\leq a \leq \delta_0 n_0-1\right\},\\
    \label{eq:set-Z-1}
    \set{Z}_1 &=& \left\{(a,b) : 0\leq b \leq \delta_1 n_1-1\right\},\\
    \set{S} &=& \left\{(a,b) : \delta_0n_0+\ \alpha (1-\delta_0)n_0 \leq a \leq n_0 - 1,\right.\nonumber\\ 
    \label{eq:set-S}
    &{}& \:\:\left.\delta_1n_1 +  \beta (1-\delta_1)n_1 \leq b \leq \delta_1n_1 +  (1-\beta)(1-\delta_1)n_1\right\},\\
    \set{T} &=& \left\{(a,b) : \delta_0n_0\leq a \leq \delta_0n_0 + \alpha(1-\delta_0)n_0-1,\right.\nonumber\\
    \label{eq:set-T}
    &{}& \:\:\left.\delta_1n_1 +  (1-\beta) (1-\delta_1)n_1 \leq b \leq n_1-1\right\}.
\end{IEEEeqnarray}

We place zeros at indices in $\set{Z}_0$ and $\set{Z}_1$, and secrets at indices $\set{S}$. 
Each of the remaining indices is assigned a uniform random mask from $\GF{q}$ (independent of other masks and the secrets). 
We use $\set{T}$ in our privacy proof.
Let $\vect{x}$ denote the resulting length-$N$ vector, which we refer to as the ``signal''.
(See Fig.~\ref{fig:SS-share-example} for a toy example.)

Let $\omega$ be a primitive $N$-th root of unity in $\GF{q}$. \proc{FastShare} computes the (fast) Fourier transform $\vect{X}$ of the signal $\vect{x}$ generated by $\omega$.\footnote{See Appendix~\ref{app:finite-field-DFT} for a brief overview of the finite field Fourier transform.} 
The coefficients of $\vect{X}$ represent shares of $\vect{s}$, i.e., $\share{\vect{s}}{i} = \vect{X}_i$. 
The details are given in Algorithm~\ref{alg:generate-shares}.

\begin{figure}[!t]
    \centering
    \includegraphics[scale=0.65]{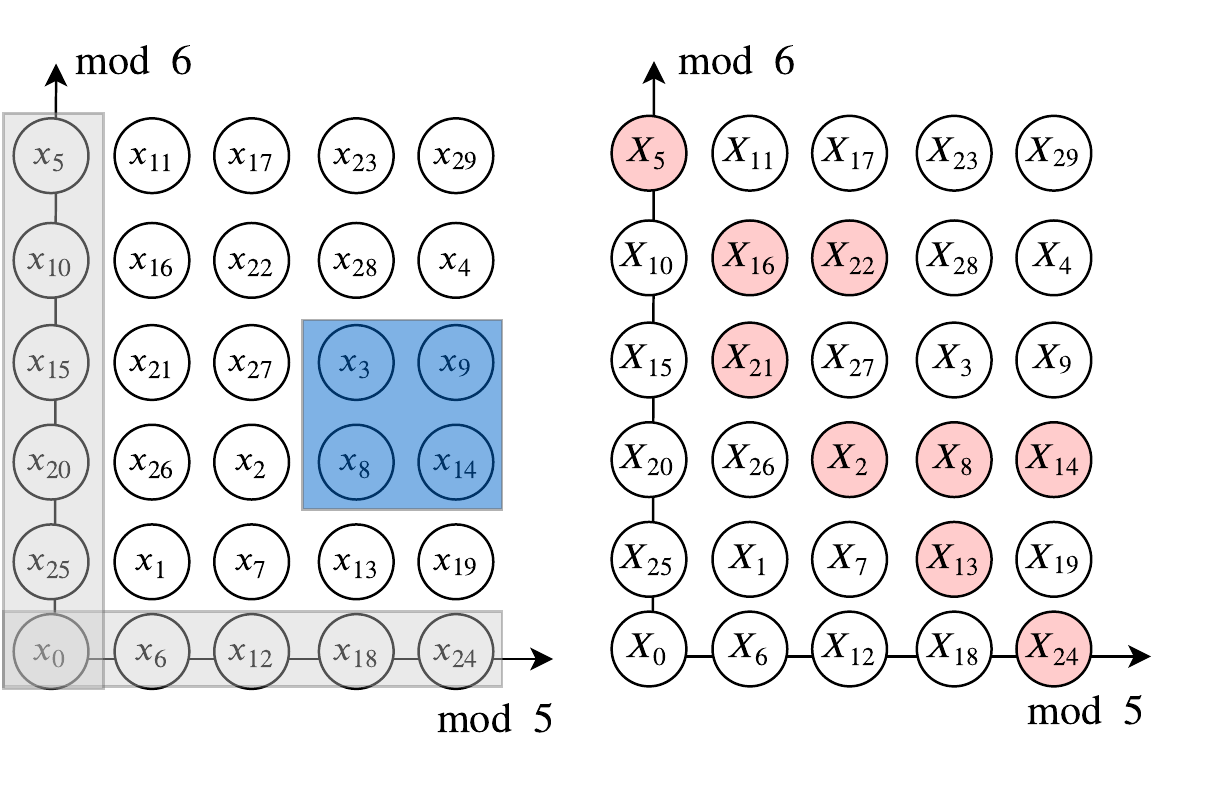}
    \caption{Example with $\ell = 4$ secrets, and co-prime integers $n_0 = 5$, $n_1 = 6$, yielding $N = n_0n_1 = 30$. Zeros are placed at gray locations, and the secrets at blue locations. Careful zero padding induces parity checks on the shares in each row and column due to \emph{aliasing}, i.e., $[ X_0+X_5+\cdots+X_{25} ; X_1+\cdots+X_{26} ; \cdots ; X_4+\cdots+X_{29}] = [0,\cdots,0]$ and $[X_0+X_6+\cdots+X_{24} ; \cdots ; X_5+\cdots+X_{29}] = [0,\cdots,0]$. Any one missing share in a row or column can be recovered using the parity-check structure. E.g., dropped out red shares  can be recovered by \emph{iteratively decoding} the missing shares in rows and columns.}
    \label{fig:SS-share-example}
\end{figure}

\begin{algorithm}[tb]
  \caption{\proc{FastShare} Secret Sharing}
  \label{alg:generate-shares}
\begin{algorithmic}
\Procedure{\textsc{FastShare}}{$\vect{s},\set{C}$}
    \State {\bfseries parameters:} finite field $\GF{q}$ of size $q$, a primitive root of unity $\omega$ in $\GF{q}$, privacy threshold $T$, dropout resilience $D$
    \State {\bfseries inputs:} secret $\vect{s}\in\GF{q}^L$, set of $ N$ clients $\set{C}$ (client-IDs)
    \State {\bfseries initialize:} sets $\set{Z}_0$, $\set{Z}_1$, $\set{S}$ as per~\eqref{eq:set-Z-0},~\eqref{eq:set-Z-0},~\eqref{eq:set-S}, respectively; an arbitrary bijection $\sigma:\set{S}\to\{0,1,\ldots,S-1\}$
      \For{$j=0$ to $n-1$}
    \If{$j\in\set{Z}_0\cup\set{Z}_1$}
        \State $x_j \leftarrow 0$
    \ElsIf{$j\in\set{S}$}
        \State $x_{j} \leftarrow \vect{s}_{\sigma(j)}$
        \Comment{$\vect{s}_i$ is the $i$-th coordinate of $\vect{s}$}
    \Else
        \State $x_j \stackrel{\$}{\leftarrow} \GF{q}$
    \EndIf
    \EndFor
    \State $X \leftarrow FFT_{\omega}(x)$
    \State {\bf Output:} $\{(i, \share{\vect{s}}{i})\}_{i\in\set{C}} \leftarrow \{(i, X_i)\}_{i=0}^{N-1}$
\EndProcedure
\State {}
\Procedure{\textsc{FastRecon}}{$\cbrack{(i,\share{\vect{s}}{i})}_{i\in\set{R}}$}
    \State {\bfseries parameters:} finite field $\GF{q}$ of size $q$, a primitive root of unity $\omega$ in $\GF{q}$, privacy threshold $T$, dropout resilience $D$, number of iterations $J$, bijection $\sigma$ and set $\set{S}$ used in \proc{FastShare}
    \State {\bfseries input:} subset of shares with client-IDs $\cbrack{(i,\share{\vect{s}}{i})}_{i\in\set{R}}$
    \For{iterations $j = 1$ to $J$}
        \For{rows $r = 0$ to $\nsubi{1}-1$ in parallel}
            \If{$r$ has fewer than $\delta_0 n_0$ missing shares}
                \State{Decode the missing share values by polynomial interpolation}
            \EndIf
        \EndFor
        \For{columns $c = 0$ to $\nsubi{0}-1$ in parallel}
            \If{$c$ has fewer than $\delta_1 n_1$ missing shares}
                \State{Decode the missing share values by polynomial interpolation}
            \EndIf
        \EndFor
    \EndFor
    \If{any missing share}
        \State{ {\bf Output:} $\perp$}
    \Else
        \State $\vect{x} \leftarrow IFFT_{\omega}(\vect{X})$
        \State {\bf Output:} $\vect{s}\leftarrow\vect{X}(\sigma(\set{S}))$
    \EndIf
\EndProcedure
\end{algorithmic}
\end{algorithm}

\subsection{\proc{FastRecon} to Reconstruct the Secrets}
\label{sec:Fast-Recon}
Let $\set{D}\subseteq\set{C}$ denote a random subset of size at most $D$. 
The reconstruction algorithm {$\{\vect{s}, \perp\} \leftarrow \proc{FastRecon}\left(\{(i,\share{\vect{s}}{i})\}_{i\in\set{C}\setminus\set{D}}\right)$} takes as input the shares corresponding to $\set{C}\setminus\set{D}$, and outputs either a set of $S$ field elements $\vect{s}$ or a special symbol $\perp$.
At a high level, \proc{FastRecon} consists of two stages. First, it \textit{iteratively} recovers all the missing shares by leveraging a {\it linear-relationship} between the shares (as described next). Second, it takes the inverse Fourier transform of the shares (i.e., the ``spectrum'') to obtain the ``signal'', which contains the secrets at appropriate locations (see Fig.~\ref{fig:SS-share}). 

The key ingredient of \proc{FastRecon} is a computationally efficient iterative algorithm, rooted in the field of coding theory, to recover the missing shares. Towards this end, we show that the shares satisfy certain ``parity-check'' constraints, which are induced by the careful placement of zeros in $\vect{x}$. (See Fig.~\ref{fig:SS-share-example} for a toy example.)

\begin{lemma}
    \label{lem:parity-checks-in-DFT}
    Let $\{X_j\}_{j=0}^{N-1}$ denote the shares produced by the \proc{FastShare} scheme for an arbitrary secret $\vect{s}$. Then, for each $i\in\{0,1\}$, for every $c \in\{0,1,\ldots,\frac{N}{n_i}\}$ and $v\in\{0,1,\ldots,\delta_i n_i-1\}$, it holds that 
    \begin{equation}
        \label{eq:shares-parity-check}
        \sum_{u=0}^{n_i-1}\omega^{-uv\frac{N}{n_i}}X_{u\frac{N}{n_i}+c} = 0.
    \end{equation}
\end{lemma}
The proof essentially follows from the subsampling and aliasing properties of the Fourier transform and is deferred to Appendix~\ref{app:proof-lemma}.

\begin{remark}
\label{rem:product-codes}
When translated in coding theory parlance, the above lemma essentially states that the shares form a codeword of a \textit{product code} with Reed-Solomon component codes~\cite{MacWilliams-Sloane:78}. In particular, when the shares are represented on a 2D-grid using the Chinese remainder theorem, each row (resp. column) forms a \textit{codeword} of a \textit{Reed-Solomon code} with block-length $n_0$ (resp. $n_1$) and dimension $(1-\delta_0)n_0$ (resp. $(1 - \delta_1)n_1$).
In other words, 
$\bar{X}_c = \begin{bmatrix}
    X_{c}&
    X_{\frac{N}{n_i}+c}&
    \cdots&
    X_{(n_i-1){\frac{N}{n_i}}+c}
\end{bmatrix}$ 
is a codeword of an $(n_i,(1-\delta_i)n_i)$ Reed-Solomon code for every $c=0,1,\ldots,N/n_i-1$.
To see this, observe that when the constraints in~\eqref{eq:shares-parity-check}, for a fixed $c$, are written in a matrix form, the resulting matrix is a Vandermonde matrix, and it is straightforward to show that $\bar{X}_c$ corresponds to the evaluations of a polynomial of degree $(1-\delta_i)n_i-1$ at $\omega^{-uN/{n_i}}$ for $u = 0,1,\ldots,n_i-1$. 
We present a comparison with the Shamir's scheme in Appendix~\ref{app:LRCs}.
\end{remark}

These parity-check constraints on the shares make it possible to iteratively recover missing shares from each row and column until all the missing shares can be recovered. We present a toy example for this in Fig.~\ref{fig:SS-share-example}. It is worth noting that this way of recovering the missing symbols of a codeword is known in coding theory as an \emph{iterative (bounded distance) decoder}~\cite{Richardson:01,Justesen:11}. 

In particular, codewords of a Reed-Solomon code with block-length $n_i$ and dimension $(1-\delta_i)n_i$ are evaluations of a polynomial of degree at most $(1-\delta_i)n_i-1$. Therefore, any $\delta_i$ fraction of erasures can be recovered via polynomial interepolation. Therefore, if a row (resp. column) has less than $\delta_0n_0$ (resp. $\delta_1n_1$) missing shares, then they can be recovered. This process is repeated for a fixed number if iterations, or until all the missing shares are recovered. 

Putting things together, \proc{FastRecon} first uses an iterative decoder to obtain the missing shares. If the peeling decoder fails, it outputs $\perp$, and declares failure. Otherwise, it takes the inverse (fast) Fourier transform of $\vect{X}$ (generated by $\omega$) to obtain $\vect{x}$. Finally, it output the coordinates of $\vect{x}$ indexed by $\set{S}$ (in the 2D-grid representation) as the secret. 
The details are given in Algorithm~\ref{alg:generate-shares}.


\subsection{Analysis of \textsc{FastShare}}
\label{sec:FastShare-analysis}
First, we analyze the security and correctness of  \proc{FastShare} in the honest-but-curious setting. We defer the proof to Appendix~\ref{app:FastShare-analysis}.
\begin{theorem}
    \label{thm:main}
   Given fractions $\delta_0, \delta_1, \alpha \in (0,1)$ and $\beta\in(0,1/2)$, \proc{FastShare} generates $N$ shares from  $S = (1-\alpha)(1-2\beta)(1-\delta_0)(1 - \delta_1) N$ secrets such that it satisfies 
   \begin{enumerate}
       \item $T$-privacy for $T = \alpha\beta(1-\delta_0)(1 - \delta_1) N$,
       \item $D$-dropout resilience for $D = (1 - (1-\delta_0)(1 - \delta_1)) \frac{N}{2}$, and
       \item linearity.
   \end{enumerate}   
\end{theorem}



For example, choosing $\alpha = 1/2$, $\beta = 1/4$, $\delta_0 = \delta_1 = 1/10$, yields $S = 0.2 N$, $T = 0.1 N$, and $D = 0.095 N$.

Next, we present the computation cost in terms of the number of (basic) arithmetic operations in $\GF{q}$.

\vspace{2pt}
\noindent\textbf{Computation Cost:} \proc{FastShare} runs in $\bigOh{N\log N}$ time, since constructing the signal takes $\bigOh{N}$ time and the Fourier transform can be computed in $\bigOh{N\log N}$ time using a Fast Fourier Transform (FFT). \proc{FastRecon} also runs in $O (N \log N)$ time, when computations are parallelized. Specifically, recovering the missing shares in a {\it row} or {\it column} of shares (when arranged in the 2D-grid) can be done in $\bigOh{n_i^2} = \bigOh{N}$ complexity by leveraging that rows and columns are Reed-Solomon codewords, which are evaluations of a degree-$(n_i-\delta_i n_i - 1)$  polynomial. Since all rows and columns can be decoded in parallel, and decoding is carried out for a constant number of iterations, the iterative decoder runs in $\bigOh{N}$ time. The second step is the inverse Fourier transform, which can be computed in $\bigOh{N\log N}$ time using an FFT.

Note that in practical FL scenarios, the number of clients typically scales up to ten thousand~\cite{Bonawitz:19:scale}. Therefore, in order to gain full parallelism, one needs to have $\sqrt{N} \approx 100$ processors/cores. In the next section, we present a variant of \proc{FastShare} that removes the requirement of parallelism for sufficiently large $N$. Moreover, this variant also allows us to achieve a more favorable trade-off between the number of shares, privacy threshold, and dropout tolerance.

\subsection{Improving the Performance for Large Number of Clients}
\label{sec:row-codes}
\proc{FastShare} ensures that when the shares are arranged in a 2D-grid using the Chinese remainder theorem, each row and column satisfies certain parity-check constraints (cf. Remark~\ref{rem:product-codes}). As we show next, when $N$ is sufficiently large, it suffices to ensure that only rows (or columns) satisfy the parity-check constraints. In this case, the \proc{FastShare} algorithm remains the same as before except for the placement of secrets, zeros, and random masks changes slightly.

\begin{figure}[!t]
    \centering
    \includegraphics[width=0.475\textwidth]{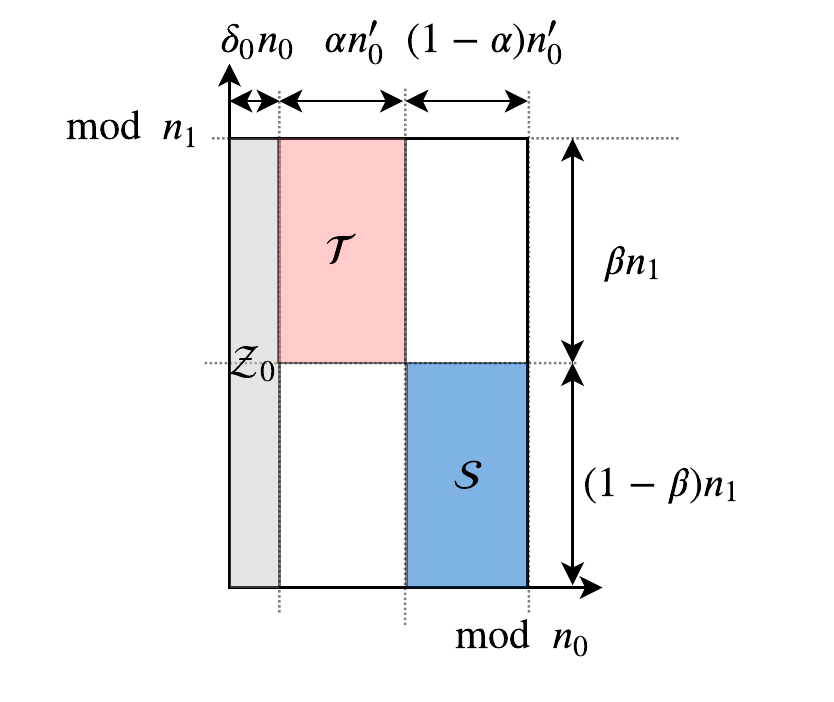}
    \caption{Indices assigned to a grid, and sets $\set{Z}$, $\set{S}$, and $\set{T}$. Here, we define $n_0' = (1-\delta_0)n_0$. 
    Set $\set{S}$ is used for secrets, $\set{Z}_0$ is used for zeros, and the remaining locations are used for random masks. Set $\set{T}$ is used in our privacy proof. 
    }
    \label{fig:2D-secrets-row-codes}
\end{figure}

Let $q$ be a power of a prime, and let $N$ be a positive integer such that $N$ divides $q-1$. Let $c \geq 1$ be a  constant. We consider $N$ of the form $N = n_0 n_1$ such that $n_0 = c \log N$. 
Given fractions $\delta_0, \alpha, \beta \in (0,1)$, we define the sets of tuples $\set{Z}_0$, $\set{S}$, and $\set{T}$ using the grid representation determined by the Chinese remainder theorem, as depicted in Fig.~\ref{fig:2D-secrets-row-codes}. We give the formal definitions below. Note that we round any rational number $x$ to $\lfloor x\rfloor$, and omit the floor sign for the sake of brevity. 
\begin{IEEEeqnarray}{rCl}
    \label{eq:set-Z-i-2}
    \set{Z}_0 &=& \left\{(a,b) : 0\leq a \leq \delta_0n_0-1\right\},\\
    \set{S} &=& \left\{(a,b) : \delta_0n_0+\ \alpha (1-\delta_0)n_0 \leq a \leq n_0 - 1,\right.\nonumber\\ 
    \label{eq:set-S-2}
    &{}& \:\:\left. 0 \leq b \leq (1-\beta) n_1 \right\},\\
    \set{T} &=& \left\{(a,b) : \delta_0n_0\leq a \leq \delta_0n_0 + \alpha(1-\delta_0)n_0-1,\right.\nonumber\\
    \label{eq:set-T-2}
    &{}& \:\:\left.(1-\beta) n_1 \leq b \leq n_1-1\right\}.
\end{IEEEeqnarray}

We construct the signal by placing zeros at indices in $\set{Z}_0$, and secrets at indices in $\set{S}$. 
Each of the remaining indices is assigned a uniform random mask from $\GF{q}$ (independent of other masks and the secrets). 
We then take the Fourier transform (generated by a primitive $N$-th root of unity in $\GF{q}$) of the signal to obtain the shares. 

Using the same arguments as in the proof of Lemma~\ref{lem:parity-checks-in-DFT}, it is straightforward to show that, for every $c \in\{0,\ldots,{n_1}-1\}$ and $v\in\{0,1,\ldots,\delta_0 n_0-1\}$, it holds that 
\begin{equation}
    \label{eq:shares-parity-check-row-codes}
    \sum_{u=0}^{n_0-1}\omega^{-uvn_1}X_{un_1+c} = 0.
\end{equation}
When translated in coding theory parlance, the above condition states that when the shares are represented in a 2D-grid using the Chinese remainder theorem, then each row forms a codeword of a Reed-Solomon code with block-length $n_0$ and dimension $(1-\delta_0)n_0$.\footnote{We note that this variant of \proc{FastShare} is related to a class of erasure codes called Locally Recoverable Codes (LRCs) (see~\cite{Gopalan:12,Prakash:12,TamoB:14}, and references therein). See Appendix~\ref{app:LRCs} for details.} 

\proc{FastRecon} first recovers the missing shares by decoding each row. If every row has less than $\delta_0$ fraction of erasures, then it recovers all the missing shares. Otherwise, it outputs $\perp$, and declares failure.  Next, it takes the inverse (fast) Fourier transform  of the shares to obtain the signal. The coordinates of the signal indexed by $\set{S}$ (in the 2D-grid representation) gives the secrets.

\vspace{4pt}
\noindent\textbf{Security and Correctness:} Given fractions $\delta_0, \alpha, \beta \in (0,1)$, this variant of \proc{FastShare} generates $N$ shares from 
\begin{equation}
    \label{eq:S-row-codes}
    S = (1-\alpha)(1-\beta)(1-\delta_0)N
\end{equation} 
secrets such that it satisfies $T$-privacy for 
\begin{equation}
    \label{eq:P-row-codes}
    T=\alpha\beta(1-\delta_0)N,
\end{equation}  
$D$-dropout-resilience for 
\begin{equation}
    \label{eq:D-row-codes}
    D = \delta_0 N
\end{equation}  
and linearity. (The proof is similar to Theorem~\ref{thm:main}, and is omitted.)

Observe that this variant achieves a better trade-off between $S$, $T$, and $D$. For instance, choosing $\delta_0 = 1/10$, $\alpha = 1/2$, and $\beta = 1/2$, yields $S = 0.225N$, $T = 0.225$, and $D = 0.1$. 

\vspace{4pt}
\noindent\textbf{Computation Cost:}
In this case, \proc{FastShare} also has $\bigOh{N\log N}$ computational cost, since constructing the signal takes $\bigOh{N}$ time and the Fourier transform can be computed in $\bigOh{N\log N}$ time using an FFT.
\proc{FastRecon} has  $\bigOh{N \log N}$ computational cost without requiring any parallelism. In particular, recovering the missing shares in a {\it row} (when arranged in the 2D-grid) can be done in $\bigOh{n_0^2} = \bigOh{\log^2 N}$ complexity by leveraging that the rows are Reed-Solomon codewords, which are evaluations of a degree-$(n_0-\delta_0 n_0 - 1)$ polynomial. Since there are $\bigOh{N/\log N}$ rows, the decoder to recover the missing shares has $\bigOh{N \log N}$ complexity. The second step is the inverse FFT, which also has $\bigOh{N\log N}$ complexity.

\section{\textsc{FastSecAgg} Based on \textsc{FastShare}}
\label{sec:FastSecAgg}

In this section, we present our proposed protocol \proc{FastSecAgg}, which allows the server to securely compute the summation of clients' model updates. We begin with a high-level overview of \proc{FastSecAgg} (see Fig.~\ref{fig:FastSecAgg}). 
Each client generates shares for its length-$L$ input (by breaking it into a sequence of $\lceil L/S \rceil$ vectors, each of length at most $S$, and treating each vector as $S$ secrets) using \textsc{FastShare}, and distributes the shares to $N$ clients. Since all the communication in FL is  mediated by the server, clients encrypt their shares before sending it to the server to prevent the server from reconstructing the secrets. Each client then decrypts and sums the shares it receives from other clients (via the server). The linearity property of \textsc{FastShare} ensures that the sum of shares is a share of the sum of secret vectors (i.e., client inputs). Each  client then sends the \textit{sum-share} to the server (as a plaintext). The server can then recover the sum of secrets (i.e., client inputs) with high probability as long as it receives the shares from a random set of clients of sufficient size. 
We note that \proc{FastSecAgg} uses secret sharing as a primitive in a standard manner, similar to several secure aggregation protocols~\cite{Ben-Or:88,Burkhart:10:sepia,Goryczka:17}.

\proc{FastSecAgg} is a three round interactive protocol. See Fig.~\ref{fig:FastSecAgg} for a high-level overview, and Algorithm~\ref{alg:FastSecAgg-algorithm} for the detailed protocol.
Recall that the model update for client $i\in\set{C}$ is assumed to be $\vect{u}_i\in\Z_R^L$, for some $R\leq q$. In practice, this can be achieved by appropriately quantizing the updates. 

\textit{Round 0} consists of generating and advertising encryption keys. Specifically, each client $i$ uses the key agreement protocol to generate a public-private key pair $(\kpk{i},\ksk{i})\leftarrow\proc{KA.gen}(pp)$, and sends their public key $\kpk{i}$ to the server. 
The server waits for at least $N-D$ clients (denoted as $\set{C}_0\subseteq\set{C}$), and forwards the received public keys to clients in $\set{C}_0$. 

\begin{figure}[!t]
    \centering
    \includegraphics[scale=0.55]{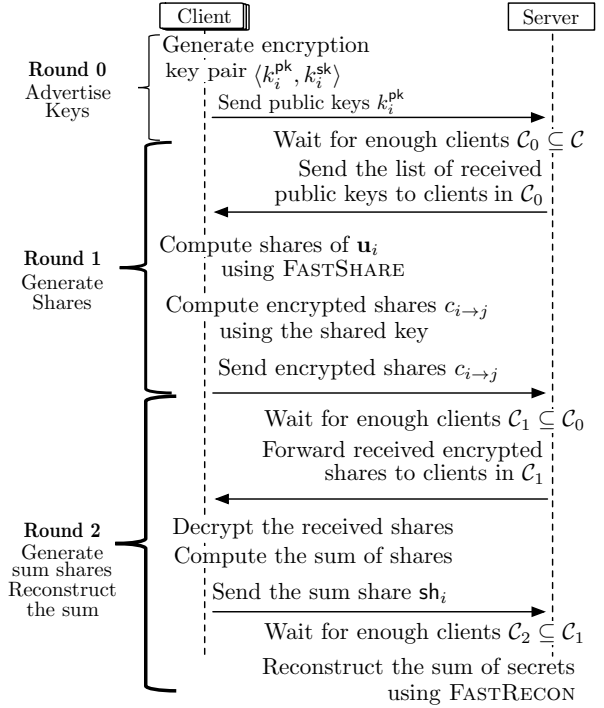}
    \caption{High-lever overview of \textsc{FastSecAgg} protocol.}
    \label{fig:FastSecAgg}
\end{figure}

\begin{algorithm*}[!h]
\caption{\textsc{FastSecAgg} Protocol}
    \begin{itemize}[topsep=0pt,itemsep=0pt,parsep=0pt,partopsep=0pt]
        \item \textbf{Parties:} Clients $1, 2, \ldots, N$ and the server 
        \item \textbf{Public Parameters:} Update length $L$, input domain $\mathbb{Z}_R$, key agreement parameter $pp\leftarrow\proc{KA.param}(\lambda)$, finite field $\GF{q}$ for \proc{FastShare} secret sharing with primitive $N$-th root of unity $\omega$
        \item \textbf{Input:} $\vect{u}_i\in\mathbb{Z}_R^L$ (for each client $i$)
        \item \textbf{Output:} $\vect{z}\in\mathbb{Z}_R^L$ (for the server)
        \vspace{2mm}
        \item \textbf{Round 0 (Advertise Keys)}\\ 
        \hspace{2mm} Client $i$:
        \begin{itemize}
            \item Generate key pairs $(\kpk{i},\ksk{i})\leftarrow\proc{KA.gen}(pp)$
            \item Send $\kpk{i}$ to the server and move to the next round
        \end{itemize}
        \hspace{2mm} Server:
        \begin{itemize}
            \item Wait for at least $N-D$  clients to respond (denote this set as $\set{C}_0\subseteq\set{C}$); otherwise, abort 
            \item Send to all clients in $\set{C}_0$ the list $\{(i,\kpk{i})\}_{i\in\set{C}_0}$, and move to the next round
        \end{itemize}
        \vspace{1mm}
        \item \textbf{Round 1} (Generate shares)\\
        \hspace{2mm} Client $i$:
        \begin{itemize}
            \item Receive the list $\{(j,\kpk{j})\}_{j\in\set{C}_0}$ from the server; Assert that $|\set{C}_0| \geq N-D$, otherwise abort
            \item Partition the input $\vect{u}_i\in\mathbb{Z}_R^L$ into $\lceil L/S \rceil$ vectors, $\vect{u}_i^1$, $\vect{u}_i^2$, $\ldots$, $\vect{u}_i^{\lceil L/S \rceil}$, such that $\vect{u}_i^{\lceil L/S \rceil}$ has length at most $S$ and all others have length $S$
            \item Compute $N$ shares by treating each $\vect{u}_i^{\ell}$ as $S$ secrets as $\{(j,\share{\vect{u}_i^{\ell}}{j})\}_{j\in\set{C}}\leftarrow\proc{FastShare}(\vect{u}_i^{\ell},\set{C})$ for $1\leq\ell\leq\lceil L/S \rceil$ (by using independent private randomness for each $\ell$); Denote client $i$'s share for client $j$ as {$\sh{i}{j} = \big(\share{\vect{u}_i^1}{j}\mid\mid \share{\vect{u}_i^2}{j}\mid\mid \cdots \mid\mid {[\vect{u}_i^{\lceil L/S\rceil}]}_{j}\big)$} 
            \item For each client $j\in\set{C}_0\setminus\{i\}$, compute encrypted share: $\esh{i}{j} \leftarrow \proc{AE.enc}(k_{i,j},i\mid\mid j\mid\mid\sh{i}{j})$, where $k_{i,j} = \proc{KA.agree}(\ksk{i},\kpk{j})$
            \item Send all the ciphertexts $\{\esh{i}{j}\}_{j\in\set{C}_0\setminus\{i\}}$ to the server by adding the addressing information $i, j$ as metadata
            \item Store all the messages received and values generated in this round and move to the next round
        \end{itemize}
        \hspace{2mm} Server:
        \begin{itemize}
            \item Wait for at least $N-D$ clients to respond (denote this set as $\set{C}_1\subseteq\set{C}_0$)
            \item Send to each client $i\in\set{C}_1$, all ciphertexts encrypted for it:
            $\cbrack{\esh{j}{i}}_{j\in\set{C}_1\setminus\{i\}}$, and move to the next round
        \end{itemize}
        \vspace{1mm}
        \item \textbf{Round 2} (Recover the aggregate update)\\
        \hspace{2mm} Client $i$:
        \begin{itemize}
            \item Receive from the server the list of ciphertexts $\cbrack{\esh{j}{i}}_{j\in\set{C}_1\setminus\{i\}}$
            \item Decrypt the ciphertext $(i'\mid\mid j'\mid\mid\sh{j}{i}) \leftarrow \mathsf{Dec}(\ksk{i},\esh{j}{i})$ for each client $j\in\set{C}_1\setminus\{i\}$, and assert that $(i=i')\: \land\:(j=j')$
            \item Compute the sum of shares over $\GF{q}$ as $\shi{i} = \sum_{j\in\set{C}_1}\sh{j}{i}$ 
            \item Send the share $\shi{i}$ to the server
        \end{itemize}
        \hspace{2mm} Server:
        \begin{itemize}
            \item Wait for at least $N-D$ clients to respond (denote this set as $\set{C}_2\subseteq\set{C}_1$)
            \item Run the reconstruction algorithm to obtain $\cbrack{\vect{z}^{\ell},\perp} \leftarrow \textsc{FastRecon}(\cbrack{(i,\shi{i}^{\ell})}_{i\in\set{C}_2})$ for $1\leq \ell \leq \lceil L/S \rceil$, where $\shi{i}^{\ell}$ is the $\ell$-th coefficient of $\shi{i}$;
            Denote $\vect{z} = [\vect{z}^1 \:\: \vect{z}^2 \:\: \cdots \:\: \vect{z}^{\lceil L/S \rceil}]$
            \item If the reconstruction algorithm returns $\perp$ for any $\ell$, then abort
            \item Output the aggregate result ${\vect{z}}$
        \end{itemize}
    \end{itemize}
    \label{alg:FastSecAgg-algorithm}
\end{algorithm*}

\textit{Round 1} consists of generating secret shares for the updates. 
Each client $i$ partitions their update $\vect{u}_i$ into $\lceil L/S \rceil$ vectors, $\vect{u}_i^1$, $\vect{u}_i^2$, $\ldots$, $\vect{u}_i^{\lceil L/S \rceil}$, such that the last vector has length at most $S$ and all others have length $S$. Treating each $\vect{u}_i^{\ell}$ as $S$ secrets, the client computes $N$ shares as \mbox{$\{(j,\share{\vect{u}_i^{\ell}}{j})\}_{j\in\set{C}}\leftarrow\proc{FastShare}(\vect{u}_i^{\ell},\set{C})$} for $1\leq\ell\leq\lceil L/S \rceil$. For simplicity, we denote client $i$'s share for client $j$ as {$\sh{i}{j} = (\share{\vect{u}_i^1}{j}\mid\mid \share{\vect{u}_i^2}{j}\mid\mid \cdots \mid\mid {[\vect{u}_i^{\lceil L/S\rceil}]}_{j})$.} 

In addition, every client receives the list of public keys from the server. Client $i$ generates a shared key $k_{i,j}\leftarrow\proc{KA.agree}(\ksk{i},\kpk{j})$ for each $j\in\set{C}_0\setminus\{i\}$, and encrypts $\sh{i}{j}$ using the shared key $k_{i,j}$ as $\esh{i}{j} \leftarrow \proc{AE.enc}\left(k_{i,j},\sh{i}{j}\right)$. The client then sends all the encrypted shares $\{\esh{i}{j}\}_{j\in\set{C}_0\setminus\{i\}}$ to the server.

The server waits for at least $N-D$ clients to respond (denoted as $\set{C}_1\subseteq\set{C}_0$).\footnote{For simplicity, we assume that a client does not drop out after initiating communication with the server, i.e., while sending or receiving messages from the server. The same assumption is also made in~\cite{Bonawitz:SecAgg:17,So:20:TurboAgg,Bell:20:SecAggPlus}. This is not a critical assumption, and the protocol and the analysis can be easily adapted if it does not hold.} Then, the server sends to each client $i\in\set{C}_1$, all ciphertexts encrypted for it: $\{\esh{j}{i}\}_{j\in\set{C}_1\setminus\{i\}}$.

\textit{Round 2} consists of generating sum-shares and reconstructing the approximate sum. Every surviving client receives the list of encrypted shares from the server. Each client $i$ then decrypts the ciphertexts $\esh{j}{i}$ using the shared key $k_{j,i}$ to obtain the shares $\sh{j}{i}$, i.e., $\sh{j}{i}\leftarrow\proc{AE.dec}(k_{j,i},\esh{j}{i})$ where $k_{j,i}\leftarrow\proc{KA.agree}(\ksk{j},\kpk{i})$. Then, each client $i$ computes the sum (over $\GF{q}$) of all the shares including its own share as: $\shi{i} = \sum_{j\in\set{C}_1}\sh{j}{i}$. 
Each client $i$ sends the sum-share $\shi{i}$ to the server (as a plaintext). 

The server waits for at least $N-D$ clients to respond (denoted as $\set{C}_2\subseteq\set{C}_1$). 
Let $\shi{i}^{\ell}$ denote the $\ell$-th coefficient of $\shi{i}$ for $1\leq\ell\leq\lceil L/S \rceil$.
The server computes $\{\vect{z}^{\ell},\perp\}\leftarrow\proc{FastRecon}\left(\{(i,\shi{i}^{\ell})\}_{i\in\set{C}_2}\right)$, for $1\leq\ell\leq\lceil L/S \rceil$. If the reconstruction fails (i.e., outputs $\perp$) for any $\ell$, then the server aborts the current FL iteration and moves to the next one. Otherwise, it outputs $\vect{z} = [\vect{z}^1 \:\: \vect{z}^2 \:\: \cdots \:\: \vect{z}^{\lceil L/S \rceil}]$.  

\section{Analysis}
\label{sec:analysis}

\subsection{Correctness and Security}
\label{sec:security-analysis}
First we state the correctness of \proc{FastSecAgg}, which essentially follows from the linearity and $D$-dropout tolerance of \proc{FastShare}. The proof is given in~\ref{app:correctness-FastSecAgg}.

\begin{theorem}[Correctness]
    \label{thm:correctness-FastSecAgg} 
    Let $\{\vect{u}_i\}_{i\in\set{C}}$ denote the client inputs for \proc{FastSecAgg}.
    If a random set of at most $D$ clients drop out during the execution of \proc{FastSecAgg}, i.e., $|\set{C}_2| \geq N - D$, then the server does not abort and obtains $\vect{z} = \sum_{i\in\set{C}_1}\vect{u}_i$ with probability at least $1 - 1/\textrm{poly}\: N$, where the probability is over the randomness in dropouts.
\end{theorem}

Next, we show that \proc{FastSecAgg} is secure  against the server colluding with up to $T$ clients in the honest-but-curious setting, irrespective of how and when clients drop out. Specifically, we prove that the joint \textit{view} of the server and any set of clients of bounded size does not reveal any information about the updates of the honest clients, besides what can be inferred from the output of the summation. 

We will consider executions of \proc{FastSecAgg} where \proc{FastShare} has privacy threshold $T$, and the underlying cryptographic primitives are instantiated with security parameter $\lambda$. We denote the server (i.e., the aggregator) as ${A}$, and the set of of $N$ clients as $\set{C}$. Clients may drop out (or, abort) at any point during the execution, and we denote with $\set{C}_i$ the subset of the clients that correctly sent their message to the server in round $i$. 
Therefore, we have $\set{C}\supseteq \set{C}_0 \supseteq \set{C}_1 \supseteq \set{C}_2$. For example, the set $\set{C}_0\setminus\set{C}_1$ are the clients that abort before sending the message to the server in Round 1, but after sending the message in Round 0.
Let $\vect{u}_i\in\mathbb{Z}_R^L$ denote the model update of client $i$ (i.e., $\vect{u}_i$ the $i$-th client's input to the secure aggregation protocol), and for any subset $\set{C}'\subseteq\set{C}$, let $\vect{u}_{\set{C}'} = \{\vect{u}_i\}_{i\in\set{C}'}$.  

In such a protocol execution, the view of a participant consists of their internal state (including their update, encryption keys, and randomness) and all messages they received from other parties. Note that the messages sent by the participant are not included in the view, as they can be determined using the other elements of their view. If a client drops out (or, aborts), it stops receiving messages and the view is not extended past the last message received.

Given any subset $\set{M}\subset\set{C}$, let $\mathsf{REAL}_{\set{M}}^{\set{C},T,\lambda}(\vect{u}_{\set{C}},\set{C}_0,\set{C}_1,\set{C}_2)$ be a random variable representing the combined views of all parties in $\set{M}\cup\{A\}$ in an execution of \proc{FastSecAgg}, where the randomness is over the internal randomness of all parties, and the randomness in the setup phase.
We show that for any such set $\set{M}$ of honest-but-curious clients of size up to $T$, the joint view of $\set{M}\cup\{A\}$ can be simulated given the inputs of the clients in $\set{M}$, and only the sum of the values of the remaining clients. 

\begin{theorem}[Security]
    \label{thm:security-FastSecAgg} 
    There exists a probabilistic polynomial time (PPT) simulator $\mathsf{SIM}$ such that for all $\set{C}$, $\vect{u}_{\set{C}}$, $\set{C}_0$, $\set{C}_1$, $\set{C}_2$, and $\set{M}$ such that $M\subset \set{C}$, $|\set{M}|\leq T$, $\set{C}\supseteq\set{C}_0\supseteq\set{C}_1\supseteq\set{C}_2$, the output of $\mathsf{SIM}$ is computationally indistiguishable from the joint view $\mathsf{REAL}_{\set{M}}^{\set{C},T,\lambda}$ of the server and the corrupted clients $\set{M}$, i.e., 
    \begin{equation}
        \label{eq:real-equal-sim}
        \mathsf{REAL}_{\set{M}}^{\set{C},T,\lambda}(\vect{u}_{\set{C}},\set{C}_0,\set{C}_1,\set{C}_2) \approx \mathsf{SIM}_{\set{M}}^{\set{C},T,\lambda}(\vect{u}_{\set{M}},\vect{z},\set{C}_0,\set{C}_1,\set{C}_2),
    \end{equation} 
    where
    \begin{equation}
        \label{eq:sum-z}
        \vect{z} =
        \begin{cases}
        \sum_{i\in\set{C}_1\setminus\set{M}}\vect{u}_i & \textrm{if}\:\:|\set{C}_1|\geq N-D,\\
        \perp & \textrm{otherwise}.
        \end{cases}
    \end{equation}
\end{theorem}



The security and correctness of \proc{FastSecAgg} critically relies on the guarantees provided by  \proc{FastShare} proved in Theorem~\ref{thm:main}. We present the proofs of the above theorem in Appendices~\ref{app:security-FastSecAgg}.

\subsection{Computation and Communication Costs}
\label{sec:costs}

We assume that the addition and multiplication operations in $\GF{q}$ are $\bigOh{1}$ each. 

\vspace{4pt}
\textbf{Computation Cost at a Client:} $\bigOh{\max\{L,N\}\log N}$. 
Each client's computation cost can be broken as computing shares using \proc{FastShare} which is $\bigOh{\lceil L/S \rceil N\log N} = \bigOh{\max\{L,N\}\log N}$; encrypting and decrypting shares, which is $\bigOh{\lceil L/S \rceil N}$; and adding the received shares, which is $\bigOh{\lceil L/S \rceil N} = \bigOh{\max\{L,N\}}$. 

\vspace{4pt}
\textbf{Communication Cost at a Client:} $\bigOh{\max\{L,N\}}$. 
Each client sends and receives $N-1$ shares (each having $\lceil L/S \rceil$ elements in $\GF{q}$), resulting in $\bigOh{\lceil L/S \rceil N} = \bigOh{\max\{L,N\}}$ communication. In addition, each client sends the sum-share consisting of $\lceil L/S \rceil$ elements in $\GF{q}$. 

\vspace{4pt}
\textbf{Computation Cost at the Server:} $\bigOh{\max\{L,N\}\log N}$.  
The server first recovers missing sum-shares using \proc{FastRecon}, each has $\lceil L/S \rceil$ elements in $\GF{q}$. This results in $\bigOh{\lceil L/S \rceil N\log N}$ complexity. The second step is the inverse FFT on $N$ shares, each has of $\lceil L/S \rceil$ elements in $\GF{q}$, resulting in $\bigOh{\lceil L/S \rceil N\log N}$ complexity.

\vspace{4pt}
\textbf{Communication Cost at the Server:} $\bigOh{N\max\{L,N\}}$.
The server communicates $N$ times of what each client communicates. 

\subsection{Comparison with Prior Works}
\label{sec:comparison}
Bonawitz et al.~\cite{Bonawitz:SecAgg:17} presented the first secure aggregation protocol \proc{SecAgg} for FL, wherein clients use a key exchange protocol to agree on pairwise additive masks to achieve privacy. \proc{SecAgg}, 
in the honest-but-curious setting,  can achieve $T = \alpha N$ for any $\alpha\in(0,1)$ and $D = N-T-1$, and provides worst-case dropout resilience, i.e., it can tolerate any $D$ users dropping out. However, \proc{SecAgg} incurs significant computation cost of $\bigOh{LN^2}$ at the server. This limits its scalability to several hundred clients as observed in \cite{Bonawitz:19:scale}.

Truex et al.~\cite{Truex:19:CCS} uses threshold homomorphic encryption and Xu et al.~\cite{Xu:19:hybridalpha} uses functional encryptionto perform secure aggregation. However, these schemes assume a trusted third party for key distribution, which typically does not hold in the FL setup.

The scheme by So et al.~\cite{So:20:TurboAgg}, which we call \proc{TurboAgg}, uses additive secret sharing for security combined with erasure codes for dropout tolerance. \proc{TurboAgg} allows $T = D = \alpha N$ for any $\alpha\in(0,1/2)$. However, it has two main drawbacks. First, it divides $N$ clients into $N/\log N$ groups, and each client in a group needs to communicate to every client in the next group. This results in per client communication cost of $\bigOh{L\log N}$. Moreover, processing in groups requires at least $\log N$ rounds. Second, it can tolerate only non-adaptive adversaries, i.e., client corruptions happen before the clients are partitioned into groups.
On the other hand, \textsc{FastSecAgg} results in $\bigOh{L}$ communication per client, runs in 3 rounds, and is robust against an adaptive adversary which can corrupt clients during the protocol execution.

Recently, Bell et al.~\cite{Bell:20:SecAggPlus} proposed an improved version of \proc{SecAgg}, which we call \proc{SecAgg+}. Their key idea is to replace the complete communication graph of \proc{SecAgg} by a sparse random graph to reduce the computation and communication costs.  \proc{FastSecAgg} achieves smaller computation cost than \proc{SecAgg+} at the server and the same (orderwise) communication cost per client as \proc{SecAgg+} when $L = \Omega(N)$. Moreover,  \proc{FastSecAgg} is robust against adaptive adversaries, whereas \proc{SecAgg+} can mitigate only non-adaptive adversaries where client corruptions happen before the protocol execution starts. On the other hand, \proc{SecAgg+} achieves smaller communication cost in absolute numbers than \proc{FastSecAgg}.

The performance comparison between \proc{FastSecAgg} and \proc{SecAgg}~\cite{Bonawitz:SecAgg:17}, \proc{TurboAgg}~\cite{So:20:TurboAgg}, and \proc{SecAgg+}~\cite{Bell:20:SecAggPlus} is summarized in Table~\ref{tbl:comparison}.

There are several recent works~\cite{Pillutla:19:Byzantine,Jaggi:20:Byzantine,So:20:Byzantine} which consider secure aggregation when a fraction of clients are Byzantine. Our focus, on the other hand, is on honest-but-curious setting, and we leave the case of Byzantine clients as a future work.

\section*{Acknowledgements}

Swanand Kadhe was introduced to the secure aggregation problem in the Decentralized Security course by Raluca Ada Popa in Fall 2019. He would like to thank Raluca for an excellent course. Swanand would like to thank Mukul Kulkarni for immensely helpful discussions and for providing feedback on drafts of this paper. This work is supported in part by the National Science Foundation grant CNS-1748692.



\bibliographystyle{ACM-Reference-Format}
\bibliography{Bib_FastSecAgg}



\appendix

\section{Finite Field Fourier Transform}
\label{app:finite-field-DFT}
Here, we review the basics of the finite field Fourier transform, for details, see, e.g.,~\cite{Pollard:71}. Let $q$ be a power of a prime, and consider a finite field $\GF{q}$ of order $q$. Let $N$ be a positive integer such that $N$ divides $(q-1)$ and $\omega$ be a primitive $N$-th root of unity in $\GF{q}$. 
The discrete Fourier transform (DFT) of length $N$ generated by $\omega$ is a mapping from $\GF{q}^N$ to $\GF{q}^N$ defined as follows. 
Let $x = [x_0\: x_1\: \dots x_{N-1}]$ be a vector over $\GF{q}$. Then, the DFT of $x$ generated by $\omega$, denoted as $DFT_{\omega}(x)$, is the vector over $\GF{q}$, $X = [X_0\: X_1\: \dots X_{N-1}]$, given by
\begin{equation}
    \label{eq:DFT}
    X_j = \sum_{i = 0}^{N - 1} \omega^{ij} x_i, \quad j = 0, 1, \dots, N-1.
\end{equation}
The inverse DFT (IDFT), denoted as $IDFT_{\omega}(X)$, is given by
\begin{equation}
    \label{eq:IDFT}
    x_i = \frac{1}{N}\sum_{j = 0}^{N - 1} \omega^{-ij} X_j, \quad i = 0, 1, \dots, N-1,
\end{equation}
where $\frac{1}{N}$ denotes the reciprocal of the sum of $N$ ones in the field $\GF{q}$.
We refer to $x$ as the ``time-domain signal'' and index $i$ as time, and $X$ as the ``frequency-domain signal'' or the ``spectrum'' and $j$ as frequency. 

If a signal is subsampled in the time-domain, its frequency components mix together, i.e., {\it alias}, in a pattern that depends on the sampling procedure. In particular, let $n$ be a positive integer that divides $N$. Let $\dsample{x}{n}$ denote the subsampled version of $x$ with period $n$, i.e., $\dsample{x}{n} = \{x_{ni} : 0\leq i\leq N/n-1\}$. Then, $(N/n)$-length DFT of $\dsample{x}{n}$, $\dsample{X}{n}$, (generated by $\omega^n$) is related to the $N$-length DFT of $x$, $X$, (generated by $\omega$) as
\begin{equation}
  \dsample{X}{n}_{j} = \frac{1}{n}\sum_{\substack{i=0\\ i\mod (N/n)=j}}^{N-1} X_i, \quad j = 0, 1, \dots, \frac{N}{n}-1,
  \label{eq:aliasing}
\end{equation}
where $1/n$ is the reciprocal of the sum of $n$ ones in $\GF{q}$.

A circular shift in the time-domain results in a phase shift in the frequency-domain. Given a signal $x$, consider its circularly shifted version $\shift{x}{1}$ defined  as $\shift{x}{1}_i = x_{(i+1)\mod N}$. Then, the DFTs of $\shift{x}{1}$ and $x$ (both generated by $\omega$) are related as $\shift{X}{1}_j = \omega^{-j} X_j$. In general, a
circular shift of $t$ results in 
\begin{equation}
  \shift{X}{t}_j = \omega^{-tj} X_j, \quad j = 0, 1, \dots, N-1.
  \label{eq:shift}
\end{equation}

\section{Proof of Lemma}
\label{app:proof-lemma}
Proof of Lemma~\ref{lem:parity-checks-in-DFT}.
For $i\in\{0,1\}$, for $v\in\{0,1,\ldots,\delta_i n_i-1\}$, let us denote $\vect{x}^{(v)(\downarrow n_i)}$ as $\vect{x}$ circularly shifted (in advance) by $v$ and then subsampled by $n_i$. That is, $\vect{x}^{(v)(\downarrow n_i)}_j = \vect{x}_{(jn_i+v)\mod n_i}$ for $j=0,1,\ldots,N/n_i-1$. Let $\vect{X}^{(v)(\downarrow n_i)}$ denote the DFT of $\vect{x}^{(v)(\downarrow n_i)}$ generated by $\omega^{n_i}$. Now, from the aliasing property~\eqref{eq:aliasing}, it holds that
\begin{equation}
    \label{eq:RS-parity-1}
    \vect{X}^{(v)(\downarrow n_i)}_c = \frac{1}{n_i}\sum_{\substack{j=0\\ j \mod (N/n_i) = c}}^{N-1}\vect{X}_j^{(v)}, 
\end{equation}
for $v = 0,1,\ldots,\delta_i n_i-1$ and $c = 0,1,\ldots,N/n_i-1.$
However, from the circular shift property~\eqref{eq:shift}, we have $\vect{X}_j^{(v)} = \omega^{-vj}\vect{X}_j$ for $j=0,1,\ldots,N-1$. Thus, we have
\begin{equation}
    \label{eq:RS-parity-2}
    \vect{X}^{(v)(\downarrow n_i)}_c = \frac{1}{n_i}\sum_{\substack{j=0\\ j \mod (N/n_i) = c}}^{N-1}\omega^{-vj}\vect{X}_j,
\end{equation}
for  $v = 0,1,\ldots,\delta_i n_i-1$, and $c=0,1,\ldots,N/n_i-1$.
Note that, by construction, $\vect{x}^{(v)(\downarrow n_i)}$ is a length-$(N/n_i)$ zero vector for $v=0,1,\ldots,\delta_i n_i-1$ for each $i\in\{0,1\}$. Therefore, $X^{(v)(\downarrow n_i)}$ is also a length-$(N/n_i)$ zero vector for $v=0,1,\ldots,\delta_i n_i-1$ for each $i\in\{0,1\}$. Hence, from~\eqref{eq:RS-parity-2}, we get
\begin{equation}
    \label{eq:RS-parity-3}
     \sum_{\substack{j=0\\ j \mod (N/n_i) = c}}^{N-1}\omega^{-vj}X_j = \sum_{u=0}^{n_i-1} \omega^{-v\left(u\frac{N}{n_i}+c\right)}X_{u\frac{N}{n_i}+c}= 0, 
\end{equation}
for $v = 0,1,\ldots,\delta_i n_i-1$, and $c=0,1,\ldots,N/n_i-1$.
Simplifying the above, we get 
\begin{equation}
    \label{eq:RS-parity-4}
    \sum_{u=0}^{n_i-1} \left(\omega^{-uv\frac{N}{n_i}}\right) X_{u\frac{N}{n_i}+c} = 0,
\end{equation}
for $u = 0,1,\ldots,n_i-1$, $v = 0,1,\ldots,\delta_i n_i-1$, and $c = 0,1,,\ldots,N/n_i-1$.


\section{Analysis of \proc{FastShare}}
\label{app:FastShare-analysis}
We first focus on the correctness, i.e.,  dropout tolerance. To prove that \proc{FastShare} has a dropout tolerance of $D$, we need to show that \proc{FastRecon} can recover the secrets from a random subset of $N-D$ shares. This is equivalent to showing that the iterative peeling decoder of \proc{FastRecon} can recover all the shares from a random subset of $N-D$ shares. Note that the shares generated by \proc{FastShare} form a codeword of a product code. From the \textit{density evolution} analysis of the iterative peeling decoder of product codes in \cite{Justesen:11,Pawar:18}, it follows that, when $D \leq (1-(1-\delta_0)(1-\delta_1))\frac{N}{2}$, the decoder recovers all the shares from a random subset of $N-D$ shares with probability at least $1 - \textrm{poly}N$. Therefore,  \proc{FastShare} has the dropout tolerance of $D = (1-(1-\delta_0)(1-\delta_1))\frac{N}{2}$.

To prove the privacy threshold of $T$, we consider the information-theoretic equivalent of the security definition~\cite{Blundo:94}. In particular, we need to show the following: for any $\set{P} \subset {C}$ such that $|\set{P}| \leq T$, it holds that $\Hcond{\vect{s}}{\{\share{\vect{s}}{i}\}_{i\in\set{P}}} = \Hp{\vect{s}}$, where $H$ denotes the Shannon entropy. For simplicity, let us denote $\share{\vect{s}}{i} = \vect{X}_i$ for all $i$. In the remainder of the proof, we denote random variables by boldface letters.


First, we show that the information-theoretic security condition is equivalent to a specific linear algebraic condition. To this end, we first observe that the vector of $N$ shares can be written as,
\begin{equation}
    \label{eq:generator}
    \vect{X} = {G}
    \begin{bmatrix}
    \vect{s}\\ \vect{m}
    \end{bmatrix},
\end{equation}
where $\vect{s}\in\GF{q}^{\ell}$ is the vector of secrets, $\vect{m}\in\GF{q}^{K-\ell}$ is a vector with each element chosen independently and uniformly from $\GF{q}$, and $G$ is a particular submatrix of the DFT matrix whose formal definition we defer later on. We then show that information theoretic security is equivalent to a particular linear algebraic condition on submatrices of $G$. To prove this condition we leverage the fact that $G$ is derived from a DFT matrix and consider an alternate representation based on the Chinese Remainder theorem (CRT). We furnish more details while formally discussing the lemmas.

Recall that $n_0$ and $n_1$ are co-prime. By the CRT, we may conclude that any number $j \in \{ 0,\cdots,N-1 \}$ can be uniquely represented in a 2D-grid as a tuple $(a,b)$ where $a = j \mod n_0$ and $b = j \mod n_1$.

First, define the set of indices
\begin{align}
    \set{S} = \{ &(p_0,p_1) : \nonumber \\
    &\delta_0 n_0 + \alpha (1 - \delta_0) n_0 \le p_0 \le n_0 - 1, \nonumber \\
    &\delta_1 n_1 + \beta (1-\delta_1)n_1 +1 \le p_1 \le n_1 - \beta n_1 (1-\delta_1) \}. \label{def:S}
\end{align}
Similarly, define
\begin{align}
    \set{Z}_0 &= \{ (p_0,p_1) : 0 \le p_0 \le \delta_0 n_0 - 1,\ 0 \le p_1 \le n_1 - 1 \}, \label{def:Z0}\\
    \set{Z}_1 &= \{ (p_0,p_1) : 0 \le p_0 \le n_0 - 1,\ 0 \le p_1 \le \delta_1 n_1 - 1 \}. \label{def:Z1}
\end{align}
For a pictorial representation of these indices, refer to Fig.~\ref{fig:2D-secrets}. These sets correspond to points in the grid, and by the CRT map back to integers in the range $\{ 0,\cdots,N-1\}$.

\begin{remark} \label{remark:1}
In \Cref{thm:main}, the number of secrets $\ell$, is chosen to be equal to $(1-\alpha)(1-2\beta)(1-\delta_0)(1-\delta_1)n_0n_1$. By design this coincides with the size of $\set{S}$.
\end{remark}

With these definition, we next describe the construction of the matrix $G$ is obtained by starting with the $N \times N$ DFT matrix, removing the columns corresponding to $\set{Z}_0 \cup \set{Z}_1$, and permuting the remaining columns so that the columns corresponding to $\set{S}$ are ordered as the first $\ell$ columns in $\set{S}$ (in arbitrary sequence).

Having defined the matrix $G$, in the following lemma we show that information theoretic security can be guaranteed by a particular rank condition satisfied by submatrices of $G$. In particular,  In addition, define $K = N - |\set{Z}_0\cup\set{Z}_1|$.

The proof is similar to Lemma 6 in \cite{Silva:11}, and thus omitted.



\begin{lemma}
    \label{lem:security}
    Let $\rand{s}\in\GF{q}^{\ell}$, $\rand{m}\in\GF{q}^{k-\ell}$, and $\rand{X} = G[\rand{s}\:\: \rand{m}]^T$ be random variables representing the secrets, random masks, and shares, respectively. Let $\rand{X}_{\set{P}}$ be an arbitrary set of shares corresponding to the indices in $\set{P}\subset\{1,2,\ldots,N\}$. Let $G_{\set{P}}$ denote the sub-matrix of $G$ corresponding to the rows indexed by $\set{P}$. Let $G_{\set{P}} = [G_1\:\: G_2]$, where $G_1$ consists of the first $\ell$ columns of $G_{\set{P}}$ and $G_2$ consists of the last $K-\ell$ columns of $G_{\set{P}}$. If $\rand{m}$ is uniform over $\GF{q}^{K-\ell}$, then it holds that
    \begin{equation}
    \label{eq:security-lemma}
    \Hp{\vect{s}} - \Hcond{\vect{s}}{\rand{X}_{\set{P}}} \leq \textrm{rank}(G_{\set{P}}) - \textrm{rank}(G_2).
    \end{equation}
\end{lemma}

Now, to prove the information-theoretic security, it suffices to prove that for any $\set{P}$ of size at most $T = \alpha\beta(1-\delta_0)(1-\delta_1)N$, $\textrm{rank} (G_2) = \textrm{rank} (G_{\set{P}})$. We prove this by showing that, for any $\set{D}$ with $|\set{P}| \leq T$, the columns of $G_1$ lie in the span of the columns of $G_2$. Note that columns of $G_1$ are the columns of the DFT matrix indexed by $\set{S}$.



\textbf{Proof of $\textrm{rank}(G_{\set{P}}) = \textrm{rank}(G_2)$}:
Assume that $\{ \omega_\partial : \partial \in \set{P} \}$ are the set of primitive elements generating the rows of $G_{\set{P}}$. For the rest of the proof we are guided by the grid representation of \cref{fig:2D-secrets} - the Chinese remainder theorem (CRT) furnishes a representation of the columns of $G_{\set{P}}$, indexed by $\{0,\cdots,n-1\}$ as points on a 2D-grid, $\{ (p_0,p_1) : p_0 \in \{ 0,\cdots,n_0-1\}, p_1 \in \{ 0,\cdots,n_1-1\} \}$.

\textit{Notation:} In the grid representation, consider the set of points $\set{T}$ in \cref{fig:2D-secrets}. Formally,
\begin{align}
    \set{T} = \{ (p_0,p_1) : \ &\delta_0 n_0 \le \ p_0 \ \le \delta_0 n_0 + \alpha (1 - \delta_0) n_0 - 1, \nonumber\\
    &n_1 - \beta (1 - \delta_1) n_1 \le \ p_1 \ \le n_1-1 \}.
\end{align}
Define $G_2 (\set{T})$ as the matrix $G_2$ only populated by the columns whose indices belong to $\set{T}$. Henceforth, we use the terminology ``points'' to refer to a column of $G_\set{P}$ in its representation as a tuple in the 2D-grid. Moreover we use the terminology ``span'' and ``rank'' of the points to indicate the span and rank of the corresponding columns. Consider a column index $p \in \{ 0,\cdots,n-1\}$ and $\Delta \in \{ 0,\cdots,n-1\}$ where $p \mapsto (p_0,p_1)$ and $\Delta \mapsto (\Delta_0,\Delta_1)$ under the CRT bijection. Then, the notation $(p_0,p_1)+(\Delta_0,\Delta_1)$ returns the point $( (p_0 + \Delta_0) \mod n_0, (p_1+\Delta_1) \mod n_1)$.

Since $G_\set{P}$ is derived from a DFT matrix, this in fact corresponds to multiplying the column corresponding to $p$ by the matrix $\text{diag} ((\omega_\partial^\Delta : \partial \in \set{P}))$. The notations (i) $p+\Delta$, (ii) $p+(\Delta_0,\Delta_1)$ are defined analogously. We also use the terminology ``shift $p$ horizontally by $\Delta_0$'' and ``shift $p$ vertically by $\Delta_1$'' to respectively denote $p+(\Delta_0,0)$ and $p+(0,\Delta_1)$. Given a set of points $\set{P}$, we overload notation and use $\set{P} + \Delta$ as the set of points $\{ p + \Delta : p \in \set{P} \}$.

Observe that $\set{P} + \Delta$ corresponds to multiplying the columns indexed by $\set{P}$ by a full-rank matrix. Then, using the structure of $G_\set{P}$ inherited from the DFT (Vandermonde) matrix structure, we get the following result immediately.
\begin{lemma} \label{remark:2}
If $p \in \text{span} (\set{P})$ for some set of points $\set{P}$, then for any $\Delta \in \{ 0,\cdots,n-1 \}$, $p+\Delta \in \text{span} (\set{P} + \Delta)$.
\end{lemma}
\begin{proof}
If $p \in \text{span} (\set{P})$ this implies that there exists weights $\{ \alpha_q : q \in \set{P}\}$ such that for each $\partial \in \set{P}$, $\omega_\partial^p = \sum\nolimits_{q \in \set{P}} \alpha_q \omega_\partial^q$. Fixing any $\partial \in \set{P}$, and multiplying both sides by $\omega_\partial^\Delta$,
\begin{equation}
    \omega_\partial^{p + \Delta} = \sum\nolimits_{q \in \set{P}} \alpha_q \omega_\partial^{q + \Delta} = \sum\nolimits_{q \in \set{P}+\Delta} \alpha_q' \omega_\partial^q.
\end{equation}
Since this is true for each $\partial \in \set{P}$ the proof follows immediately.
\end{proof}
\begin{figure}
    \centering
    \includegraphics[width=0.37\textwidth]{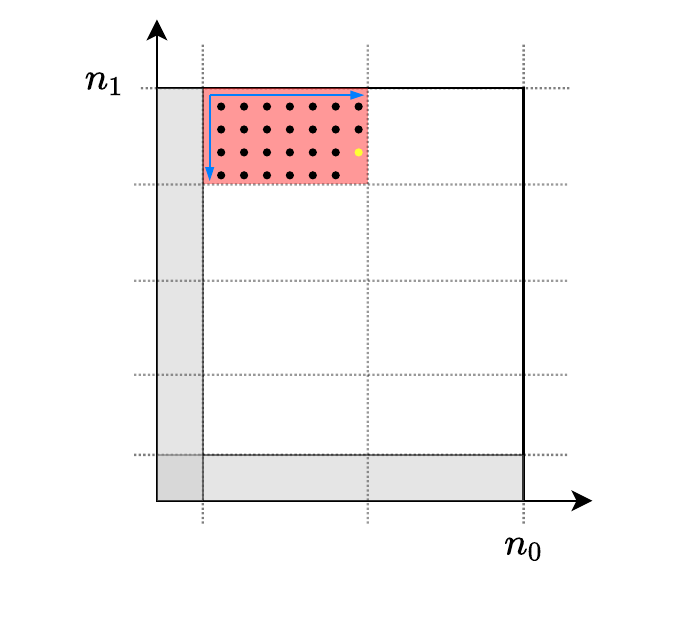}
    \caption{Enumerating the columns in the bottom-right sequence until the first column (yellow) is reached which lies in the span of the previous columns}
    \label{fig:2}
\end{figure}
By construction of the set, observe that the number of points in $\set{T}$ equals $\alpha \beta (1 - \delta_0) (1 - \delta_1) n_0 n_1$. Moreover, recalling the choice of the privacy threshold $T$ in \Cref{thm:main}, observe that $|\set{T}| = T$. By this fact, we have that $\rank (G_{\set{P}}) \leq T$. This implies two possibilities:
\begin{description}
    \item[Case I.] If $\rank (G_2 (\set{T})) = \rank (G_{\set{P}})$, then the remaining columns in $G_1$ must lie in the span of $G_2 (\set{T})$, and the proof concludes.
    \item[Case II.] If not, then the columns of $G_2(\set{T})$ must have at least one column dependent on the others.
\end{description}

We can identify one such column by the following procedure: starting at the top left corner of $\set{T}$ in the grid, sequentially collect the points in $\set{T}$ in a top-to-bottom, left-to-right sequence. Stop at the first point $y = (y_0,y_1)$ which lies in the span of the previously collected points (by assumption, such a point must exist as we are in Case II). \Cref{fig:2} illustrates this procedure where we collect the black points sequentially until the yellow point is reached which is the first point that lies in the span of the previously collected points. Let $\set{B}$ be defined as the set of black points which are enumerated until $y$ is found. By definition, this implies that there exist (not necessarily unique) weights $\{ \alpha_b \}_{b \in \set{B}}$ such that,
\begin{equation}
    \omega^y = \sum\nolimits_{b \in \set{B}} \alpha_b \omega^b
\end{equation}
This implies that $y - (0,1) \in \text{span} (\set{B} - (0,1))$. Define $\set{Y}$ as the set of points,
\begin{equation}
    \set{Y} = \{ (y_0,\theta) :  y_1 - n_1(1-\delta_1)(1-\beta)+1 \le \theta \le y_1 \}.
\end{equation}
Note that the set of points $\set{Y}$ implicitly depends on the location of point $y$. Next, we define $\set{L}$ and $\set{C}$  to be the set of points,
\begin{alignat}{3}
    \set{L} &=&& \cup_{k=0}^{n_0-y_0-1} (\set{Y} + (k,0)) \label{def:L}\\
    \set{C} &=&& \{ (p_0,p_1) : \nonumber\\
    & &&n_0\delta_0 \le p_0 \le y_0,  n_1 \delta_1 \le p_1 \le n_1 - 1 \} \mathbin{\big\backslash} \set{Y} \label{def:C}.
\end{alignat}
In addition, we define $\set{Q}$ to be the set of points,
\begin{alignat}{3} \label{def:Q}
    \set{Q} &=&& \{ (p_0,p_1) : \nonumber\\
    & &&n_0\delta_0 \le p_0 \le n_0-1,  n_1 \delta_1 \le p_1 \le n_1 - 1 \} \mathbin{\big\backslash} \set{L}.
\end{alignat}
Pictorially these points are represented in Figs~\ref{fig:4} and~\ref{fig:5}.
\begin{claim} \label{claim:1}
$\set{Y} \subseteq \text{span} (\set{C})$.
\end{claim}
\begin{proof}
We prove this result by induction sequentially iterating over all the points from top to bottom in $\set{Y}$. Clearly $\set{B} \subseteq \set{C}$, therefore the induction holds for the first point in $\set{Y}$, which is $y = (y_0,y_1)$.

Suppose for some $k \ge 1$ that for each $t \le k-1$ the point $y-(0,t)$ lies in $\text{span} (\set{C})$. Then we show that $y- (0,k) \in \text{span} (\set{C})$ as long as $k \ge y_1 - n_1(1-\delta_1)(1-\beta)+1$. Indeed, observe first that,
\begin{equation} \label{eq:1203102}
    y - (0,k) \in \text{span} (\set{B} - (0,k)).
\end{equation}
The key observation is that $\set{B} - (0,k)$ is always $\subseteq \set{C} \cup \{ y - (0,t) : 0 \le t \le k-1 \}$ (unless $k > y_1 - n_1(1-\delta_1)(1-\beta)+1$, in which case $\set{B} - (0,k)$ shifts far enough down that it includes points in $\set{Z}_1$ - this has null intersection with $\set{C}$ and $\set{Y}$ so the assertion is clearly false). For a pictorial presentation of this fact, refer to Fig.~\ref{fig:6}. By the induction hypothesis, for each $t \le k-1$, $y - (0,t) \in \text{span} (\set{C})$. Therefore, $\text{span} (\set{B} - (0,k)) \subseteq \text{span} (\set{C})$. Plugging into \cref{eq:1203102} completes the proof of the claim.
\end{proof}

\begin{claim} \label{claim:2}
$\set{L} \subseteq \text{span} (\set{Q})$.
\end{claim}
\begin{proof}
Recall that $\set{L} = \cup_{k=0}^{n_0-y_0-1} (\set{Y} + (k,0))$. We follow a similar proof by induction strategy as \Cref{claim:1} to result in the assertion. In particular, for $k = 0$ the statement follows directly from the fact that $\set{C} \subseteq \set{Q}$ and using \Cref{claim:1} to claim that $\set{Y} \subseteq \text{span} (\set{C})$. Assuming the induction hypothesis that for all $0 \le t \le k-1$, $\set{Y} + (t,0) \subseteq \text{span} (\set{Q})$ we show that $\set{Y} + (k,0) \subseteq \text{span} (\set{Q})$ as long as $k \le n_0-y_0-1$. In particular, observe that from \Cref{claim:1},
\begin{equation} \label{eq:12312111}
    \set{Y} + (0,k) \subseteq \text{span} (\set{C} + (0,k)).
\end{equation}
Next observe that unless $k > n_0-y_0-1$,
\begin{equation} \label{eq:10211}
    \set{C} + (k,0) \subseteq \set{Q} \cup \{ \set{Y} + (t,0) : 0 \le t \le k-1 \}
\end{equation}
if $k > n_0-y_0-1$, then $\set{C} + (k,0)$ shifts far enough that it wraps around and include points in $\set{Z}_0$ - this has null intersection with $\set{Q}$ and $\set{Y} + (t,0)$ for any $t \le n_0-y_0-1$ and the assertion becomes false. A pictorial representation of this fact is provided in Fig.~\ref{fig:7}. By the induction hypothesis for each $0 \le t \le k-1$, $\set{Y} + (t,0) \subseteq \text{span} (\set{Q})$. Combining this fact with \cref{eq:12312111,eq:10211} implies that $\set{Y} + (0,k) \subseteq \text{span} (\set{Q})$ and completes the induction step for $k$.
\end{proof}

Recall from \cref{def:L} that $\set{L}$ is defined as the set of points, $\cup_{k=0}^{n_0-y_0-1} (\set{Y} + (k,0))$. More explicitly, noting that the set of points $\set{Y}$ is defined as $\{ (y_0,\theta) : y_1-n_1(1-\delta_1)(1-\beta) + 1 \le \theta \le y_1 \}$,
\begin{align}
    \set{L} = \{ (p_0,p_1) &: \ y_0 \le p_0 \le n_0-1, \nonumber\\
    &y_1-n_1(1-\delta_1)(1-\beta) + 1 \le p_1 \le y_1\}.
\end{align}

\begin{claim} \label{claim:3}
The set of points $\set{S} \subseteq \set{L} $ irrespective of the location of the point $y = (y_0,y_1)$ (note that the set of points $\set{L}$ is a function of $y_0$ and $y_1$).
\end{claim}
\begin{proof}
Before furnishing the details of the proof, note that a pictorial representation of this statement is provided in Fig.~\ref{fig:7}. First recall that $y = (y_0,y_1)$ is a point in the set $\set{T}$. Therefore, $\delta_0 n_0 \le y_0 \le \delta_0 n_0 + \alpha (1-\delta_0) n_0 - 1$ and $n_1 - \beta (1-\delta_1) n_1 \le y_1 \le n_1-1$. In particular, this means that for any $y$,
\begin{align}
    \left\{ \rule{0cm}{1cm}\right.
    \begin{split}
        &(p_0,p_1) :\\
        &\delta_0 n_0 + \alpha (1-\delta_0) n_0 - 1 \le p_0 \le n_0-1,\\
        &y_1-n_1(1-\delta_1)(1-\beta) + 1 \le p_1 \le y_1
    \end{split} \left.\rule{0cm}{1cm}\right\} \subseteq \set{L}
\end{align}
Furthermore, we see that (i) $y_1 \ge n_1 - \beta (1-\delta_1)n_1$, and (ii) $y_1-n_1(1-\delta_1)(1-\beta) + 1 \le n_1 - n_1 (1-\delta_1)(1-\beta) = n_1\delta_1 + \beta(1-\delta_1)n_1$. Therefore, we have that,
\begin{align}
    \left\{ \rule{0cm}{1cm}\right.
    \begin{split}
        &(p_0,p_1) :\\
        &\delta_0 n_0 + \alpha (1-\delta_0) n_0 - 1 \le p_0 \le n_0-1,\\
        &n_1\delta_1 + \beta(1-\delta_1)n_1 \le p_1 \le n_1 - \beta (1-\delta_1)n_1
    \end{split} \left.\rule{0cm}{1cm}\right\} \subseteq \set{L}
\end{align}
The LHS is exactly the definition of $\set{S}$ in \cref{def:S}.
\end{proof}

From Claims~\ref{claim:2} and~\ref{claim:3} and the definition of $\set{L}$ in \cref{def:L}, we see that $\set{S} \subseteq \text{span} (\set{Q})$. In particular this implies that the each column indexed by points in $\set{S}$ can be expressed as a linear combination of some set of remaining columns that do not belong to $\set{Z}_0$ or $\set{Z}_1$. This implies that for any choice of $\set{P}$ (note that we do not specify a particular choice of $\set{P}$ in the proof as long as it has size $\le T$) that $\rank (G_2) = \rank (G_\set{P})$.

Moreover, we do not specify the particular choice of $\set{P}$ in the proof, so it holds for all sets of size $\le T$.

\begin{figure}[!t]
    \centering
    \includegraphics[width=0.37\textwidth]{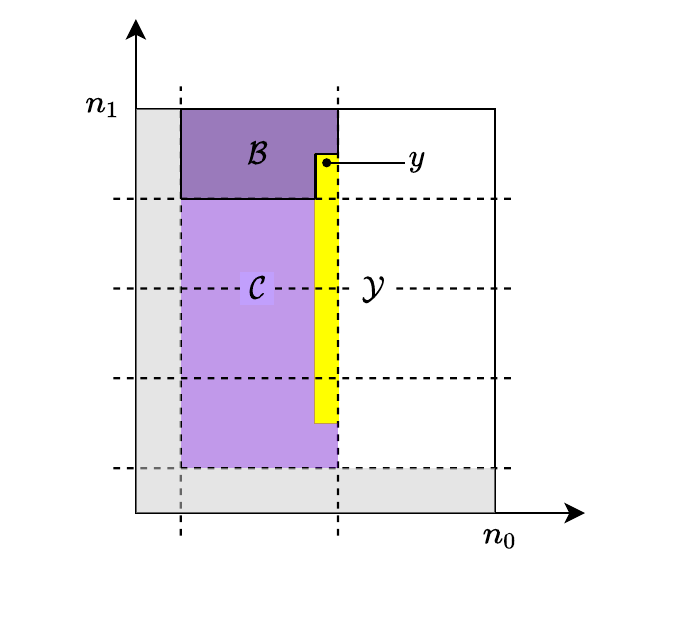}
    \caption{Definition of $\set{B},\set{Y}$ and $\set{C}$}
    \label{fig:4}
\end{figure}

\begin{figure}[!t]
    \centering
    \includegraphics[width=0.37\textwidth]{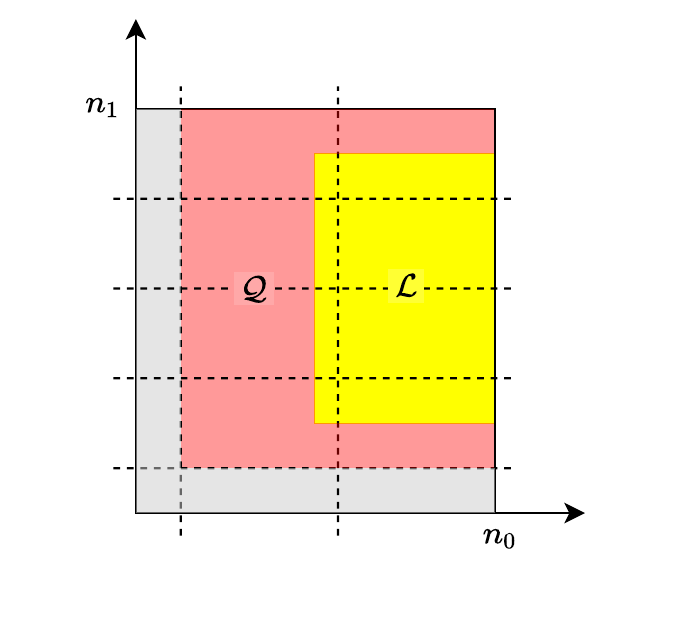}
    \caption{Definition of $\set{Q}$ and $\set{L}$}
    \label{fig:5}
\end{figure}

\begin{figure}[!t]
    \centering
    \includegraphics[width=0.37\textwidth]{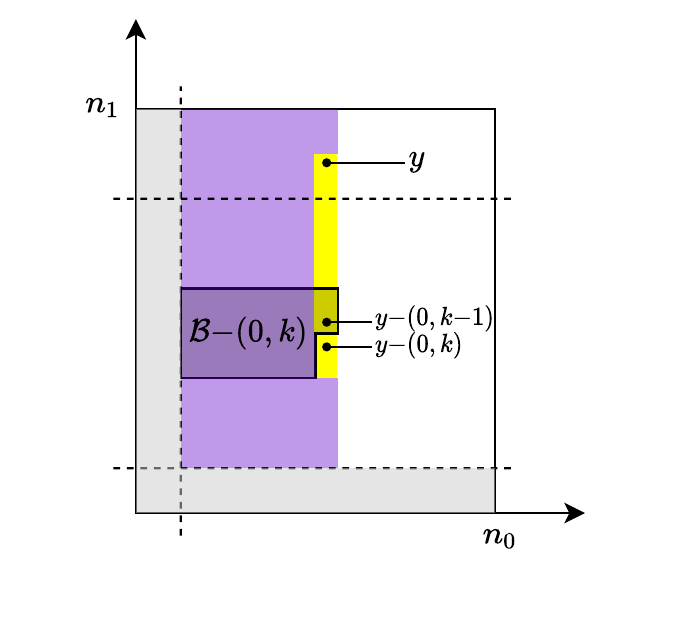}
    \caption{Showing that $\set{B} - (0,k)$ is a subset of the intersection of $\set{C}$ and $\{ y - (0,t) : 0 \le t \le k-1 \}$.}
    \label{fig:6}
\end{figure}

\begin{figure}[!t]
    \centering
    \includegraphics[width=0.37\textwidth]{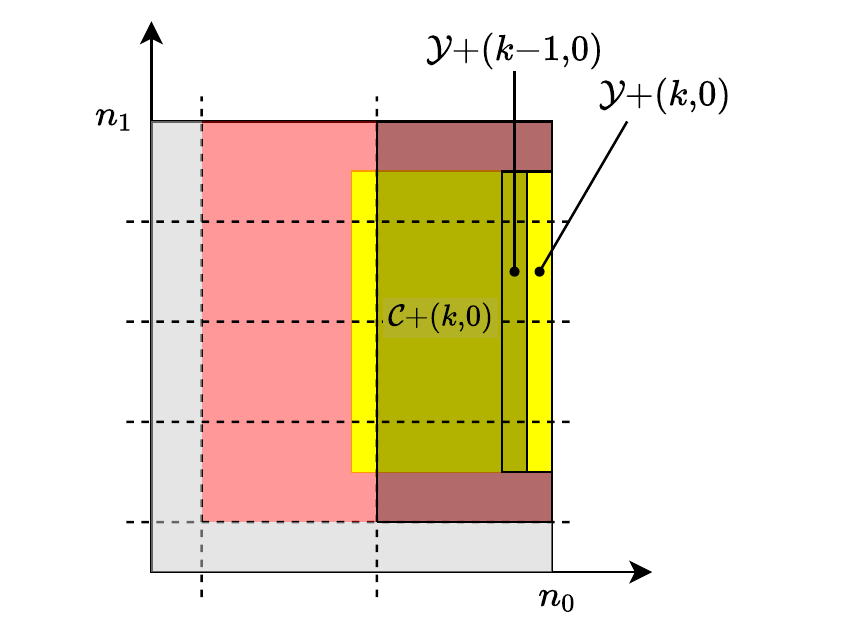}
    \caption{Showing that $\set{C}+(k,0)$ is a subset of the intersection of $\set{Q}$ and $\bigcup_{t=0}^{k-1} (\set{Y} + (t,0))$.}
    \label{fig:7}
\end{figure}

\begin{figure}[!t]
    \centering
    \includegraphics[width=0.37\textwidth]{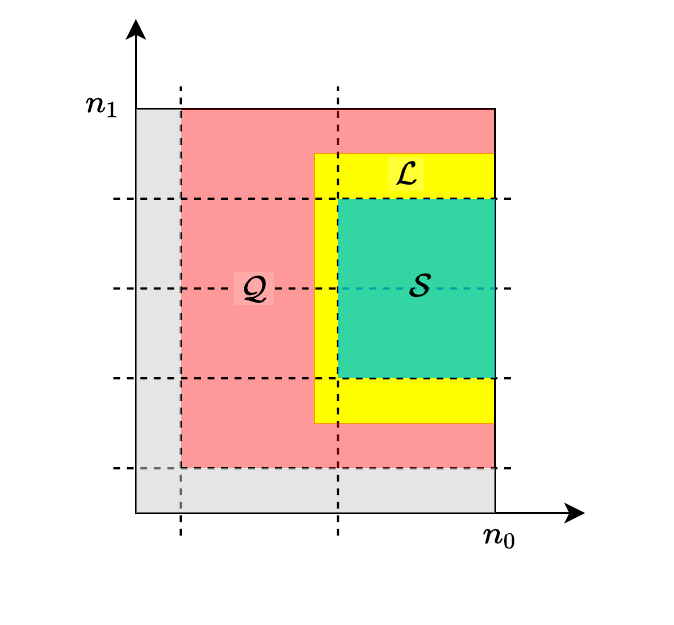}
    \caption{The $\set{S}$ is indeed a subset of $\set{L}$. For the particular choice of $y$ this figure shows that the inclusion is true. More generally we in construct $\set{S}$ as the intersection of $\set{L}$ over all the possible locations of $y \in \set{T}$.}
    \label{fig:8}
\end{figure}

\section{Correctness of \proc{FastSecAgg}}
\label{app:correctness-FastSecAgg}
Correctness essentially follows from the correctness of key agreement and authenticated encryption protocols together with the linearity and $D$-dropout tolerance of \proc{FastShare}.

Specifically, using the correctness of key agreement and authenticated encryption, it is straightforward to show that, in Round 2, each client $i\in\set{C}_1$ computes and sends to the server the sum of shares it receives in Round 1 as follows
\begin{equation}
\label{eq:sum-share-1-2}
\shi{i} = \left(\sum_{j\in\set{C}_1}[\vect{u}_j^1]_i\mid\mid\cdots\mid\mid\sum_{j\in\set{C}_1}[\vect{u}_j^{\lceil L/S \rceil}]_i \right).
\end{equation} 
The linearity property of \proc{FastShare} ensures that the sum of shares is a share of the sum of secret vectors (i.e., client inputs). In other words, by linearity of \proc{FastShare}, it holds that
\begin{equation}
\label{eq:sum-share-2}
\shi{i} = \left(\left[\sum_{j\in\set{C}_1}\vect{u}_j^1\right]_i\mid\mid\cdots\mid\mid\left[\sum_{j\in\set{C}_1}\vect{u}_j^{\lceil L/S \rceil}\right]_i \right).
\end{equation} 
Recall that $\shi{i}^{\ell}$ denotes the $\ell$-th coefficient of $\shi{i}$. Therefore, from~\eqref{eq:sum-share-2}, it holds, for $1\leq\ell\leq\lceil L/S \rceil$, that
\begin{equation}
\label{eq:sum-share-ell-2}
\shi{i}^{\ell} = \left[\sum_{j\in\set{C}_1}\vect{u}_j^{\ell}\right]_i.
\end{equation} 

Let $\set{C}_2$ be a random subset of at least $N-D$ clients that survive in Round 2, i.e., clients in $\set{C}_2$ send their sum-shares to the server. 
Now, from~\eqref{eq:sum-share-ell-2} and the $D$-dropout tolerance of \proc{FastShare}, it holds that $\proc{FastRecon}\left(\{(i,\shi{i}^{\ell})\}_{i\in\set{C}_2}\right)=\sum_{j\in\set{C}_1}\vect{u}_j^{\ell}$ with probability at least $1 - 1/\textrm{poly}\: N$ for all $1\leq\ell\leq\lceil L/S \rceil$. 
Note that we do not need to use a union bound here because if \proc{FastRecon} succeeds (i.e., does not output $\perp$) for $\ell = 1$, then it succeeds for each $1\leq\ell\leq\lceil L/S \rceil$, since the locations of missing indices (among $N$) of shares for every $1\leq\ell\leq\lceil L/S \rceil$ are the same. 

Hence, we get 
$$[\vect{z}^1\:\:\vect{z}^2\:\:\cdots\:\:\vect{Z}^{\lceil L/S\rceil}] = \left[\sum_{j\in\set{C}_1}\vect{u}_j^{1}\:\:\sum_{j\in\set{C}_1}\vect{u}_j^{2}\:\:\cdots\:\:\sum_{j\in\set{C}_1}\vect{u}_j^{\lceil L/S\rceil}\right]$$ 
with probability at least $1 - 1/\textrm{poly}\: N$. In other words, $\vect{z} = \sum_{i\in\set{C}_1}\vect{u}_i$ with probability at least $1 - 1/\textrm{poly}\: N$. This concludes the proof.

\section{Security of \proc{FastSecAgg}}
\label{app:security-FastSecAgg}

In \proc{FastSecAgg}, each client $i$ first partitions their input $\vect{u}_i$ to $\lceil L/S \rceil$ vectors $\vect{u}_i^1$, $\vect{u}_i^2$, $\ldots$, $\vect{u}_i^{\lceil L/S \rceil}$, each of length at most $S$, and computes shares for each $\vect{u}_i^{\ell}$, $1\leq\ell\leq\lceil L/S \rceil$. 
In other words, we have
\begin{equation}
    \label{eq:shares-u-i-l}
    \left\{\left(j,\share{\vect{u}_i^{\ell}}{j}\right)\right\}_{j\in\set{C}}\leftarrow\proc{FastShare}(\vect{u}_i^{\ell},\set{C}),
\end{equation}
where independent private randomness is used for each $1\leq\ell\leq \lceil {L}/{S} \rceil$.
Let us denote the set of shares that client $i$ generates for client $j$ as $\share{\vect{u}_i}{j}$, i.e., we have
\begin{equation}
    \label{eq:share-i-to-j}
    \share{\vect{u}_i}{j} = \left\{\share{\vect{u}_i^1}{j}, \share{\vect{u}_i^2}{j}, \ldots,  {[\vect{u}_i^{\lceil L/S\rceil}]}_{j}\right\}.
\end{equation}

It is straightforward to show that for any set of up to $T$ clients $\set{P}\subset\set{C}$, the shares $\{\share{\vect{u}_i}{j}\}_{j\in\set{P}}$ reveal no information about $\vect{u}_i$. 

\begin{lemma}
\label{lem:privacy-for-set-of-shares}
For every $\vect{u}_i,\vect{v}_i\in\GF{q}^L$, for any $\set{P}\subset\set{C}$ such that $|\set{P}|\leq T$, the distribution of $\{\share{\vect{u}_i}{j}\}_{j\in\set{P}}$ is identical to that of $\{\share{\vect{v}_i}{j}\}_{j\in\set{P}}$.
\end{lemma}
\begin{proof}
Here we treat $\vect{u}_i$ to be a random variable (with arbitrary distribution), and $\{\share{\vect{u}_i}{j}\}_{j\in\set{P}}$ to be a conditional random variable given a realization of $\vect{u}_i$. 
By slightly abusing the notation for simplicity, denote $\vect{u}_i$ as a set $\vect{u}_i = \{\vect{u}_i^1,\vect{u}_i^2,\ldots,\vect{u}_i^{\lceil L/S \rceil}\}$ (instead of a vector).
Further, for simplicity, define
$$\share{\vect{u}_i^{\ell}}{\set{P}} = 
\{\share{\vect{u}_i}{j}\}_{j\in\set{P}};\quad
\share{\vect{u}_i^{\ell}}{\set{P}} = \left\{\share{\vect{u}_i^{\ell}}{j}\right\}_{j\in\set{P}}, \:\: 1\leq\ell\leq\lceil L/S \rceil.$$

Now, observe that, given $\vect{u}_i^{\ell}$, the distribution of $\share{\vect{u}_i^{\ell}}{\set{P}}$ is conditionally independent of $\vect{u}_i\setminus\{\vect{u}_i^{\ell}\}$ and $\share{\vect{u}_i}{\set{P}}\setminus\left\{\share{\vect{u}_i^{\ell}}{\set{P}}\right\}$.
This is because, for the $\ell$-th instantiation,  $1\leq\ell\leq\lceil L/S\rceil$, \proc{FastShare} takes only $\vect{u}_i^{\ell}$ as its input and uses independent private randomness.
Therefore, the distribution of $\share{\vect{u}_i}{\set{P}}$ factors into the product of the distributions of $\share{\vect{u}_i^{\ell}}{\set{P}}$. The proof then follows by applying the $T$-privacy property for each instantiation of $\proc{FastShare}$.
\end{proof}

We prove Theorem~\ref{thm:security-FastSecAgg} by a standard hybrid argument. We will present a sequence of hybrids starting that start from the real execution and transition to the simulated execution where each two consecutive hybrids are computationally indistinguishable.

\noindent$\textbf{Hybrid}_0:$ This random variable is $\mathsf{REAL}_{\set{M}}^{\set{C},T,\lambda}(\vect{u}_{\set{C}},\set{C}_0,\set{C}_1,\set{C}_2)$, the joint view of the parites $\set{M}$  in the real execution of the protocol.

\noindent$\textbf{Hybrid}_1:$ In this hybrid, we change the behavior of honest clients in
$\set{C}_1\setminus\set{M}$ so that instead of using $\proc{KA.agree}(\ksk{i},\kpk{j})$ to encrypt and decrypt messages, we run the key agreement simulator $\mathsf{Sim}_{KA}(s_{i,j},\kpk{j})$, where $s_{i,j}$ is chosen uniformly at random.  
The security of the key agreement protocol guarantees that this hybrid is  indistinguishable from the previous one.

\noindent$\textbf{Hybrid}_2:$ In this hybrid, for every client $i\in\set{C}_1\setminus\set{M}$, we replace the shares of $\vect{u}_i$ sent to other honest clients in Round 1 with zeros, which the adversary observes encrypted as $\esh{i}{j}$. Since only the contents of the ciphertexts are changed, the IND-CPA security of the encryption scheme guarantees that this hybrid is indistinguishable from the previous one. 

\noindent$\textbf{Hybrid}_3:$ In this hybrid, for every client $i\in\set{C}_1\setminus\set{M}$, we replace the shares of $\vect{u}_i$ sent to the corrupt clients in $\set{M}$ in Round 1 with shares of $\vect{v}_i$, which are chosen as follows depending on $\vect{z}$. If $\vect{z} = \perp$, then $\{\vect{v}_i\}_{i\in\set{C}_1\setminus\set{M}}$ are chosen uniformly at random. Otherwise, $\{\vect{v}_i\}_{i\in\set{C}_1\setminus\set{M}}$ are chosen uniformly at random subject to $\sum_{i\in\set{C}_1\setminus\set{M}}\vect{v}_i = \vect{z} \left(= \sum_{i\in\set{C}_1\setminus\set{M}}\vect{u}_i\right)$. The joint view of corrupt parties contains only $|\set{M}|\leq T$ shares of each $\vect{v}_i$. From Lemma~\ref{lem:privacy-for-set-of-shares}, it follows that this hybrid is identically distributed to the previous one.

If $\vect{z} = \perp$, then we do not need to consider the further hybrids, and let $\mathsf{SIM}$ is defined to sample from $\textbf{Hybrid}_3$. This distribution can be computed from the inputs $\vect{z}$, $\set{C}_0$, and $\set{C}_1$. In the following hybrids, we assume $\vect{z}\ne \perp$.

\noindent$\textbf{Hybrid}_4:$ Partition $\vect{z}$ into $\lceil L/S \rceil$ vectors, $\vect{z}^1$, $\vect{z}^2$, $\ldots$, $\vect{z}^{\lceil L/S \rceil}$, each of length at most $S$. In this hybrid, for every client $i\in\set{C}_2\setminus\set{M}$ and each $1\leq\ell\leq\lceil L/S \rceil$, we replace the share
$\shi{i}^{\ell}$ with the $i$-th share of $\vect{z}^{\ell} + \sum_{j\in\set{M}}\vect{u}_j^{\ell}$, i.e., $\share{\vect{z}^{\ell}+\sum_{j\in\set{M}}\vect{u}_j^{\ell}}{i}$. 
Since $\vect{z} = \sum_{i\in\set{C}_1\setminus\set{M}}\vect{u}_i$, from~\eqref{eq:sum-share-2} and~\eqref{eq:sum-share-ell-2}, it follows that the distribution of $\left\{\share{\vect{z}+\sum_{j\in\set{M}}\vect{u}_j}{i}\right\}_{i\in\set{C}_2}$ is identical to that of $\{\shi{i}\}_{i\in\set{C}_2}$. Therefore, this hybrid is identically distributed to the previous one.

We define a PPT simulator $\mathsf{SIM}$ to sample from the distribution described in the last hybrid.
This distribution can be computed from the inputs $\vect{z}$, $\vect{u}_{\set{M}}$, $\set{C}_0$, $\set{C}_1$, and $\set{C}_2$. The argument above proves that the output of the simulator is computationally indistinguishable from the output of $\mathsf{REAL}$, which concludes the proof.

\section{Comparison of \proc{FastShare} with Shamir's Scheme and Relation to Locally Recoverable Codes}
\label{app:LRCs}
\noindent \textbf{Comparison with the Shamir's scheme:} A multi-secret sharing variant of the Shamir's scheme can be constructed as follows (see, e.g.,~\cite{Franklin-Yung:92}). We consider the case of generating $N$ shares (each in $\GF{q}$, $q\geq N$) from $S$ secrets (each in $\GF{q}$) such that the privacy threshold is $T$ and dropout tolerance is $N-S-T$. Construct a polynomial of degree-$(S+T-1)$ with the secrets and $T$ random masks (each in $\GF{q}$) as coefficients. Its evaluations on $N$ distinct non-zero points in $\GF{q}$ yield the $N$ shares for the Shamir's scheme. 

Note that it is possible to recover the secrets from any $S+T$ shares via polynomial interpolation. Although there exist fast polynomial interpolation algorithms that can interpolate a degree-$n$ polynomial in $\bigOh{n\log^2 n}$ time (see, e.g.,~\cite{Horowitz:72,Borodin:74,Gathen:book:99}), the actual complexity of these algorithms is typically large due to huge constants hidden by the big-Oh notation as noted in ~\cite{Borodin:74,Gathen:book:99,Luby:01}. 
This limits the scalability of such fast interpolation algorithms.

Since the shares generated by Shamir are evaluations of a degree-$(S+T-1)$ polynomial on $N$ distinct non-zero points, it is easy to see that the shares form a codeword of a Reed-Solomon code of block-length $N$ and dimension $S+T$. On the other hand, \proc{FastShare} constructs a polynomial of degree-$(N-1)$, and evaluates it on primitive $N$-th roots of unity. The shares generated by \proc{FastShare} form a codeword of a product code with Reed-Solomon codes as component codes (see Remark~\ref{rem:product-codes}). 

Suppose we choose primitive $N$-th roots of unity as the evaluation points for the Shamir's multi-secret sharing scheme. Then, the Shamir's scheme essentially computes the Fourier transform of the ``signal'' consisting of $S$ secrets, concatenated with $T$ random masks followed by $N-S-T$ zeros. In contrast, in case of \proc{FastShare}, we construct the signal by judiciously placing zeros, which ensures the product code structure on the shares.

\vspace{4pt}
\noindent\textbf{Relation to LRCs:}
The variant of \proc{FastShare} described in~\ref{sec:row-codes} is related to a class of erasure codes called Locally Recoverable Codes (LRCs) (see~\cite{Gopalan:12,Prakash:12,TamoB:14}, and references therein). Specifically, consider LRCs with $(\ell,r)$ locality~\cite{Prakash:12} in which coordinates can be partitioned into several subsets of cardinality $\ell$ such that the coordinates in each subset form a maximum distance separable (MDS) code of length $\ell$ and dimension $r$. 
An efficient construction of such LRCs is presented in~\cite[Section V-C]{TamoB:14}. 
Note that the shares generated by \proc{FastShare} (as described in~\ref{sec:row-codes}) form codewods of an LRC with $(n_0,(1-\delta_0)n_0)$ locality. 
Moreover, it is not difficult to show that \proc{FastShare} as an {\it evaluation code} is isomorphic to the LRC constructed in~\cite[Section V-C]{TamoB:14} when the {\it good polynomial} is $x^{\ell}$ 
(see~\cite[Sections III-A]{TamoB:14} for  the definition of a good polynomial) and the evaluation points are primitive $N$-th roots of unity.

\end{document}